%% file: cooperating.tex
\pdfoutput=1
\documentclass[12pt,reqno]{article}

\usepackage{mathpazo}
\usepackage{setspace}
\usepackage{amsmath,amsfonts,amssymb,amsthm,mathtools,bm}
\usepackage{graphicx,float}
\usepackage{placeins,needspace}
\usepackage{tikz}
\usetikzlibrary{arrows.meta}
\usepackage{geometry}
\usepackage{titlesec}
\usepackage{enumitem}
\usepackage{xcolor}
\usepackage{mdframed}
\usepackage{microtype}
\usepackage[bottom]{footmisc}
\usepackage[authoryear,round]{natbib}
\usepackage{hyperref}
\hypersetup{
  colorlinks=true,
  linkcolor=red!60!black,
  citecolor=blue!75!black,
  urlcolor=blue!60!black
}
\usepackage[nameinlink,noabbrev,sort,capitalise]{cleveref}
\usepackage{comment}

\crefformat{equation}{#2(#1)#3}
\crefrangeformat{equation}{#3(#1)#4--#5(#2)#6}
\crefmultiformat{equation}{#2(#1)#3}{ and #2(#1)#3}{, #2(#1)#3}{, and #2(#1)#3}

\newcommand{\de}{\mathop{}\!\mathrm{d}}
\newcommand{\E}{\mathbb E}

\newcommand{\ind}{\mathbf 1}
\newcommand{\MPE}{\ensuremath{\mathsf{MPE}}}
\newcommand{\SPE}{\ensuremath{\mathsf{SPE}}}

\theoremstyle{definition}

\newtheorem{definition}{Definition}
\newtheorem{lemma}{Lemma}
\newtheorem{proposition}{Proposition}
\newtheorem{theorem}{Theorem}

\newtheorem{example}{Example}

\titlespacing*{\paragraph}{0pt}{1.25ex plus 1ex minus .2ex}{0.5em}

\titleformat{\section}
  {\bfseries\centering\MakeUppercase}
  {\thesection}{0.5em}{}[]
\titleformat{\subsection}[runin]
  {\normalfont\bfseries}
  {\thesubsection}{0.5em}{\addperiod}[]
\newcommand{\addperiod}[1]{#1.}
\allowdisplaybreaks

\begin{document}

\title{\textbf{\textls[-50]{\textsc{Cooperating against Catastrophe}}}
\thanks{Fudenberg: MIT Department of Economics \url{drew.fudenberg@gmail.com}; Koh: Columbia University Department of Economics \url{andrew.koh@columbia.edu}. We thank Allan Dafoe, Daniel Kokotajlo, and Phil Trammel for helpful conversations and National Science Foundation  grant SES-2417162 for financial support.}
}
\author{\makebox[.25\linewidth]{Drew Fudenberg}\\ \normalsize{MIT} \\  
\and \makebox[.25\linewidth]{Andrew Koh}\\ \normalsize{Columbia}}
\date{\today}
\maketitle
\begin{abstract}
We study a continuous-time game in which two firms choose how quickly to advance their capabilities while an exogenous safety threshold advances at a fixed rate. When the capability frontier (max of firms' capabilities) exceeds the safety threshold, all firms are exposed to common disaster that arrives at a hazard rate increasing in the capability-safety difference. We characterize Markov perfect equilibria in terms of the disaster risk, how quickly safety advances, and flow payoffs: when they are low, only racing is an equilibrium; when they are intermediate, racing and pacing coexist; when they are high, only pacing survives. 
Across all subgame perfect equilibria, low risk induces perpetual racing while high risk rules it out and we bound the probability of disaster across all SPE. If internal capabilities are hidden with a fixed lag until deployment, % with a lag and are hidden until deployed, 
it is harder to sustain pacing which highlights the importance of transparency about internal capabilities. 

\end{abstract}

\thispagestyle{empty} \vspace{-1em}

\newpage \setcounter{page}{1}
\onehalfspacing

\section{Introduction}

In September 2026, Dario Amodei, CEO of Anthropic, advocated for \emph{pacing the frontier}: slowing the development of AI capabilities to give AI safety research and institutions time to catch up \citep{amodei2026pacing}. Other leaders of American AI firms, including Sam Altman, Demis Hassabis, and Elon Musk, soon voiced their support, and more than 1,300 computer scientists signed an open letter urging the US government to support an international effort to develop the tools needed to pace frontier AI development \citep{pacingthefrontier_2026}. At the same time, critics have warned that government intervention could primarily protect incumbent firms and their rents \citep{doctorow2026reverse,horowitz2026opensource}. One obstacle to pacing the frontier is that each firm may reason that if it slows down, a rival will press ahead.\footnote{Versions of this argument have been advanced by, among others, \cite{andreessen2023ai,aschenbrenner2024situational,cowen2026doomsday}.} Yet strategic rivalry can also produce mutual restraint when the parties share an interest in avoiding a bad outcome \citep{schelling1960strategy}. When can self-interested, profit-maximizing firms slow AI research to a safe pace without government intervention?

To answer this question, we analyze a continuous-time game where two firms each choose a development speed at every date, from zero (pausing) to one (racing at maximum speed).\footnote{We interpret players as firms that maximize profits, but we think our results might also apply to US-China AI competition.} Each firm's technological capability (e.g., that of its AI model) is the integral of past speeds;  the capability frontier is the higher of the two firms' capabilities. An exogenous safety threshold (e.g., the state of AI alignment research, or institutional preparedness) advances at a constant, exogenously fixed rate strictly between zero and one. The capability frontier's distance above the safety threshold  measures how far the leader has outpaced the safety benchmark.

When the frontier exceeds safety, firms are exposed to a positive risk of disaster that permanently terminates both firms' profits.\footnote{We assume that this is the only effect of disaster in the firms. If they internalize any social losses, it would be easier to support pacing the frontier and reduce the set of parameters where both firms always race.} The hazard governing disaster risk is a nondecreasing and concave function of the distance above safety, and might include a discrete jump at the boundary. Before disaster, each firm's flow profit increases with its own capability and decreases with its rival's. Each firm's expected payoff is the present discounted value of flow profits, weighted by the probability of survival up to each date. There is no direct cost of development.

Most of our paper focuses on the Markov Perfect Equilibria (MPE) of this game,\footnote{Here the state is the safety threshold and each firm's capability.} and their relation to  two  simple strategy profiles that capture the possibilities of unrestrained development and endogenous restraint. The \emph{racing profile} prescribes undertaking R\&D at the maximum speed at every state. The \emph{pacing profile} prescribes that a firm that is strictly behind (`follower') always races, and the firm ahead (`leader') follows an optimal stopping rule: race to a chosen distance above safety, and then pause until either safety or its rival catches up. When firms are tied, they both pause when their capabilities are not too high above the safety threshold, and race beyond it.

Our main result  partitions the parameter space into three regions according to the effective hazard relative to two cutoff values determined by the safety speed and profit parameters: 
    When the hazard rate is low,  racing profile is an MPE and no MPE outcome sustains permanent joint pacing from a tied position on safety. Under an additional discounting condition, perpetual joint racing is the unique subgame-perfect equilibrium (SPE) outcome from a safety tie.     At intermediate hazard rates, both the racing and pacing profiles are MPE.  Finally, with high hazard the pacing profile is an MPE and the racing profile is not.

We show faster safety progress, a higher hazard profile, and a lower discount rate each favor pacing over racing in equilibrium. Specifically, each of these changes weakly expands the parameter region in which the pacing profile is an MPE and weakly shrinks the region in which the racing profile is an MPE. Under a linear hazard profile, the equilibrium boundaries depend strictly and monotonically on both the safety speed and the discount rate, and they converge to strictly positive limits as the safety speed approaches the maximum. Consequently, safety progress arbitrarily close to the maximum development speed does not eliminate the racing equilibrium.% Numerical examples illustrate that the same level of racing-path disaster exposure can be consistent with different equilibrium outcomes, because faster safety improves the pacing alternative even when it leaves racing exposure unchanged. They also show that a steeper hazard profile can reduce equilibrium disaster probability if it induces the leader to accommodate rather than race, so the behavioral response can outweigh the direct increase in exposure.

Next, we analyze the outcomes of all subgame-perfect equilibria (SPE). At low risk and under an additional discounting condition, perpetual joint racing is the unique SPE outcome. We also bound disaster probability at every finite date across all SPE. To do so, we construct a recursive payoff guarantee: a firm can unilaterally pause and benefit from its rival's incentive to avoid disaster, then use the resulting guarantee to strengthen the same argument at earlier histories. Doing so gives a lower bound on how much cooperation can be sustained across all SPEs. We also show, for more general functional forms, how the threat of mutual racing in the future can supply incentives to pace the frontier in the present. 

Finally, to reflect the gap between AI firms' deployed (public) and internal (often private) capabilities \citep{metr2026frontier}, we suppose that internal capabilities become observable and generate profits and risk only after a deployment lag. We show how this lag shapes the sustainability of pacing: a longer lag allows a firm to build a larger lead before its rival can respond, making a deviation from pacing more attractive.

\textbf{Related work.} Our work connects to several strands of existing work.

\emph{R\&D races and preemption.}
A large literature studies dynamic competition to develop a new technology. The early models of \cite{loury1979market}, \citet{lee1980market}, and \citet{dasgupta1980uncertainty} relate market structure to the intensity and speed of R\&D through a single investment choice;   \citet{reinganum1982dynamic} introduces commitment to  time-varying open-loop  strategies. \citet{fudenberg1983preemption}, \citet{harris1985perfect} and \citet{ harris1987racing} study preemption, leapfrogging, and the discouragement of laggards in models where firms observe and respond to their rivals' progress.  In our model, unlike these, research 
generates a common disaster hazard.

\emph{AI races.}  \citet{Dafoe2018}, \citet{askell2019role},\citet{BostromDafoeFlynn2020}, and \citet{ZwetslootDafoe2019}  analyze AI competition without developing formal models.  \citet{ArmstrongBostromShulman2016}  and \citet{de2026agi} analyze how incentives to finish first
can induce developers to sacrifice safety in static games where firms make once-and-for-all decisions, and \citet{EmeryXuParkTrager2024} studies how incomplete information about heterogeneous capabilities affects investment in safety. \citet{hendrycks2025superintelligence} and \citet{aifutures2026ai2040} outline how US-China AI cooperation might be sustained via the threat of sabotaging or destroying a rival's data centers.  \citet{StaffordTragerDafoe2022} studies a dynamic race to implement a risky technology, where waiting permits further learning. \citet{fudenberg2026racing}
studies a  two-firm  R\&D race where firms observe and respond to information about their opponent's progress and stopping is irreversible; it focuses on the effects of the speed of monitoring and beliefs about opponent's rationality.

\emph{Continuous time games.} The earliest paper of continuous-time games where  each player controls a rate such as  \citet{Isaacs1965},  assumed that equilibrium strategies are continuous.  \citet{fudenberg1983capital} studied discontinuous investment strategies in a model of strategic deterrence using the concept of \emph{state space equilibrium}, which was subsequently generalized by   \citet{MaskinTirole2001} to games where the state is not a primitive of the model. 
Continuous-time R\&D races with a continuous effort control go back to
\citet{reinganum1982dynamic} and \citet{harris1987racing}.  The closest formal
relatives are the strategic experimentation games of \citet{BoltonHarris1999}
and \citet{KellerRadyCripps2005}: there too each player chooses an intensity
in $[0,1]$ at every date, the objective is linear in the intensity, and the
authors must say which feedback rules generate a well-defined path.  The
general difficulty of defining strategies on continuous-time histories is
discussed by \citet{SimonStinchcombe1989} and \citet{BerginMacLeod1993}.

\section{Model}

\subsection{Development, safety, and disaster risk}

Time is continuous, $t\geq0$.  There are two firms $i\in\{1,2\}$.  Firm $i$
chooses a development speed $a_{it}\in[0,1]$, which it can change at any
time.  Its total capability  is
$
A_{it}=A_{i0}+\int_0^t a_{is}\,\de s. \footnotemark$ We call speed one \emph{racing} and speed zero \emph{pausing}.  Pausing
leaves the firm's capability unchanged, and the firm can resume development at any time.\footnotetext{We will later impose measurability conditions that ensure this integral is well defined.}

\noindent
The \emph{capability frontier} $A_t$ is the higher of the two firms' capabilities: $A_t:=\max\{A_{1t},A_{2t}\}$. The \emph{safety threshold} $X_t$ advances at a constant speed $g\in(0,1)$:
$X_t=X_0+gt.$ 
The capability frontier's distance above safety is
$
D_t:=A_t-X_t.$ Thus when $D>0$ frontier capabilities are ahead of the safety threshold and vice versa for $D<0$. If one firm has a strict lead and develops at speed $a_L$, then
$\dot D=a_L-g$.

The disaster risk $\lambda$ depends on distance above safety, with
\begin{equation*}\tag{RISK}\label{eq:risk}
\begin{aligned}
\lambda(D)&=\begin{cases}
0,&D\leq0,\\
\lambda_0+\phi(D),&D>0,
\end{cases}
\end{aligned}
\end{equation*}
where $\phi$
is continuous on $[0,\infty)$, nondecreasing, concave, and twice
continuously differentiable on $(0,\infty)$, with $\phi(0)=0$. 
Since $\lambda(D) = 0$, pacing exactly along the safety threshold incurs no risk. 
Note that risk is generated by the capability frontier's position relative to safety (what \cite{TrammellAschenbrenner2026} call ``state risk'') so existing capabilities that remain strictly above safety continue to generate risk even during a pause.

Disaster permanently ends both firms' future profits.  The \emph{disaster time} is the random variable $\tau_D$ with survival function
$
\Pr(\tau_D > t) = \exp(- \int_0^t \lambda(D_s)\,\de s)$ 
which, for each capability path $(A_{1t},A_{2t})$, specifies a law for the disaster as a function of the distance of the capability frontier from safety.

\subsection{Payoffs}
To simplify the analysis, we use a functional form for the payoffs that implies the only source of non-stationarity in the model is the evolution of the safety threshold. Specifically, before disaster, firm $i$ receives flow payoff
$
\pi_i(A_i,A_j)=\exp\{\alpha A_i-\beta A_j\},$where $\alpha>\beta>0$ and there is no direct cost of development.
Own development raises profit at rate $\alpha$; rival development lowers
it at rate $\beta$.  Since $\alpha>\beta$, advancing both capabilities
equally raises both firms' profits.  Firm $i$'s expected payoff is
\begin{equation*}
V_i
=\E\!\left[
\int_0^{\tau_D}e^{-rt}
e^{\alpha A_{it}-\beta A_{jt}}\,\de t
\right]
=\int_0^\infty e^{-rt-\Lambda_t}\,
e^{\alpha A_{it}-\beta A_{jt}}\,\de t, where
\end{equation*}
the equality uses  the fact, explained
next, that the speed paths are deterministic while the game is ongoing.   We assume $r>\alpha-\beta$ throughout,
so a tied pair has finite discounted payoffs even if both firms race
forever without disaster.

\subsection{Histories, information, and strategies}
\label{sec:strategies}
The state $S_t$ records both firms' capabilities and the safety threshold, $S_t=(A_{1t},A_{2t},X_t)\in\mathcal S:=\mathbb R^2\times\mathbb R$.  Since $a_{it}\in[0,1]$, each firm's capability path is absolutely
continuous, nondecreasing, and $1$-Lipschitz. Conversely, every
absolutely continuous path with derivative in $[0,1]$ almost everywhere
can be written as
$
A_{it}=A_{i0}+\int_0^t a_{is}\,\de s
$ for a measurable speed path $a_i:[0,\infty)\to[0,1]$, unique almost
everywhere. 

For each initial state
$S=(A_1,A_2,X)\in\mathcal S$, let
\[
\Omega_S(S):=
\left\{
(S_t)_{t\geq0}\in\mathcal C([0,\infty),\mathcal S):
\begin{array}{l}
S_0=S,\\
X_t=X^0+gt\text{ for every }t\geq0,\\
A_1,A_2\text{ are absolutely continuous} \\ 
\dot A_{1t}, \dot A_{2t}\in[0,1]\text{ a.e.}
\end{array}
\right\}.
\]
Let
$
\Omega_S:=\bigcup_{S\in\mathcal S}\Omega_S(S)$ 
be the collection of feasible state paths from arbitrary initial
states, and let $\Omega:=\Omega_S\times[0,\infty]$ be the full outcome
space, whose second coordinate is the disaster time $\tau_D$. The information in the public history up to time $t$ is given by the filtration 
$
\mathcal F_t:=\sigma(S_{t'},\ \ind\{\tau_D\leq t'\}: t'\leq t)
$; we write
$\mathcal F:=(\mathcal F_t)_{t\geq0}$ for how this information evolves over time.

\begin{definition}[Strategies]\label{def:strategy}
A \emph{pure strategy} for firm $i$ is a map
$\sigma_i:[0,\infty)\times\Omega\to[0,1]$ 
that is progressively
measurable with respect to $\mathcal F$,\footnotemark{}
where $\sigma_i(t,\omega)$ is the prescribed speed at date $t$
given the public history in $\omega$ up to that date.
\end{definition}
\footnotetext{That is, for every $T<\infty$, its
restriction to $[0,T]\times\Omega$ is
$\mathcal B([0,T])\otimes\mathcal F_T$-measurable.}

Progressive measurability prevents a firm from using future information
and ensures measurability in time.   

A pure \emph{Markov strategy} depends on the history only through the
current state: it is a Borel measurable map $\sigma_i:\mathcal S\to[0,1]$
assigning a prescribed speed to each state $S=(A_1,A_2,X)$.\footnote{We are abusing notation here, because Markov strategies depend only on the most recent state.} 

\subsection{Paths and payoffs at discontinuities}
\label{sec:generalized-paths}

A Markov strategy prescribes a speed at each state, and the state evolves
according to these prescribed speeds.  The following example shows that 
strategies need not determine play.

\begin{example}[Non-unique paths]\label{ex:multiple-paths}
Suppose each firm races except when its capability equals safety,
where it paces:
\[
\sigma_i(S)=\begin{cases}
g,&A_i=X,\\
1,&A_i\neq X.
\end{cases}
\]
Start both capabilities at zero, with $X_0>0$. Both firms race until
the time $T$ where $A_{it}=T=X_t$. Because investment at precisely time t has no direct impact on the intergral defining the  evolution of states,  both the path $A_{it}=X_t=T+g(t-T)$ and the path $A_{it}=t$ are feasible paths.  \end{example}

We address this ambiguity by (i) specifying
which paths are admissible under a strategy profile and following deviations; and (ii) assigning payoffs over those paths.  We retain prescribed boundary
actions and evaluate disaster risk from the
realized frontier.  The verification below bounds every admissible
deviation path and exhibits one common path that attains both firms'
equilibrium payoffs.

For any pair of Markov policies $\sigma=(\sigma_1,\sigma_2)$, define
\begin{equation*}\tag{K}\label{eq:krasovskii}
K[\sigma](S)
=\bigcap_{\varepsilon>0}
\overline{\operatorname{co}}
\Big\{(\sigma_1(S'),\sigma_2(S')):S'\in\mathcal S,\;
                  \|S'-S\|<\varepsilon\Big\}.
\end{equation*}
The set $K[\sigma](S)\subseteq[0,1]^2$ contains the joint speeds allowed
at state $S$.  Here $\overline{\operatorname{co}}$ is the closed convex hull.  In words,
we allow averages of policy values arbitrarily close to the state,
including values prescribed exactly at that state.  This is the \emph{regularized map} generated by $\sigma$.\footnote{More precisely, it is the Krasovskii regularization. \citet{KrasovskiiSubbotin1988}; \citet[Sections 1--2]{BivasEtAl2020}
compares it with the  Filippov regularization. More recent work in differential games e.g., \citet{JaakkolaWagener2026}
 use a regularization based on Filippov solutions, and evaluate
profiles by their best admissible payoffs.}
For a single Markov policy $\sigma_j$, we write $K[\sigma_j](S)$ for
its allowed speeds under regularization. This set is defined by
\cref{eq:krasovskii} with the scalar policy $\sigma_j$ in place of the
pair, still evaluated at the full state $S$.

An \emph{admissible path}from $S$ is an absolutely continuous state path
starting at $S$ and satisfying, almost everywhere,
\begin{equation*}\tag{A}\label{eq:generalized-path}
(\dot A_{1t},\dot A_{2t})\in K[\sigma](S_t)
\quad \text{and} \quad \dot X_t=g.
\end{equation*}
Let $\mathcal{P}(\sigma; S)$ denote the set of paths starting at $S$ satisfying \cref{eq:generalized-path}; since this condition is a differential inclusion, $\mathcal{P}(\sigma; S)$ is  its solution set from $S$, with speeds given by the path derivatives.
The regularized map has nonempty compact convex values, is upper
semicontinuous, and is bounded, 
so Filippov's existence theorem for differential inclusions (Theorem 2.1 in \citet{Filippov1988})
%the Kakutani–Filippov Existence Theorem for differential inclusions (see e.g. Theorem 3 in  \citet{cellina2005view}) 
supplies a global path from every
state. 

\paragraph{Risk and payoffs.}
 Write $V_i[(S_t)_{t\geq0}]$ for the original discounted payoff evaluated
on a state path.  Define
$
V_i(\sigma,S)
=\sup_{(S_t)_{t\geq0}\in\mathcal P(\sigma;S)}V_i[(S_t)_{t\geq0}].$ 
Each firm values a profile by its best admissible continuation. Note that the suprema need not be attained.  Payoffs may be
$+\infty$ for arbitrary profiles, but our assumption that $r>\alpha-\beta$ ensures finiteness for tied paths. The equilibria constructed below have finite payoffs.

\paragraph{Deviations and equilibrium.}
A \emph{deviation} replaces firm $i$'s strategy by an alternative strategy
$\sigma_i'$, holding fixed the
opponent's strategy $\sigma_j$ fixed.
Write $(\sigma_i',\sigma_j)$ for the profile with only firm $i$'s
strategy replaced. When discussing continuation play, a \emph{history} is a path of states at which disaster has not occurred and we write
$h_t=(S_s)_{0\leq s\leq t}$ for such a history. It is \emph{feasible}
if its state path is absolutely continuous and,   $\dot A_{is}\in[0,1]$ for $i = 1,2$ and $\dot X_s=g$ almost everywhere on
$[0,t]$. 

Write $\mathcal P(\sigma;h_t)$ for the
admissible continuations after $h_t$. Define the continuation payoff from history $h_t$ under the strategy profile $\sigma$ by 
\[
V_i(\sigma,h_t)
:=\sup_{(S_s)\in\mathcal P(\sigma;h_t)}
\int_t^\infty e^{-r(s-t)-(\Lambda_s-\Lambda_t)}
\pi_i(S_s)\,\de s.
\]

\begin{definition}[Markov-perfect equilibrium]
\label{def:generalized-mpe}\label{def:mpe}
A \emph{Markov profile} $\sigma$ is a Markov-perfect equilibrium (\MPE) if,
after every feasible history $h_t$ starting from any initial state, neither firm can gain from a deviation: for each firm $i$ and
every alternative strategy $\sigma_i'$,
$V_i((\sigma_i',\sigma_j),h_t)\leq V_i(\sigma,h_t).$ An \emph{MPE outcome} is an admissible path under an MPE profile that
attains both firms' assigned continuation payoffs.
 \end{definition}

Although deviations can depend on the full history of past play, a Markov best response always exists under our path and payoff rules.

\begin{lemma}[Markov best responses]\label{lem:markov-best-response}
If the opponent uses a Markov strategy, then there is a Markov
strategy that is a best response.
Consequently, a Markov profile $\sigma$ is an \MPE{} if and only if,
for every state $S$, firm $i$, and alternative Markov strategy $\sigma_i'$,
$
V_i((\sigma_i',\sigma_j),S)\leq V_i(\sigma,S).$ 
\end{lemma}
This lemma would be immediate in a discrete-time version of this game, but requires a proof here because of regularization; the proof is in \cref{sec:proof}

\section{Equilibrium}
\label{sec:equilibrium}\label{sec:main}

\subsection{Racing and pacing strategies}
\label{sec:tied-values}

We first normalize each firm's continuation payoff by its current flow
profit. At a tie $A_1=A_2=A$, both firms' current flow profits equal
$e^{(\alpha-\beta)A}$. Let $P(D)$ be normalized payoffs: payoff divided by initial flow profit
$e^{(\alpha-\beta)A}$, when 
firms start tied at distance $D$ above safety pause until safety catches up, and pace along the safety threshold forever, or start below safety and race until they reach the safety threshold and pace forever. 

At $D=0$ they pace, so 
$
P(0)=\frac1{r-g(\alpha-\beta)}$. For $D>0$, the wait lasts $D/g$. Profits and disaster risk continue
during the pause, so the payoff is
\begin{align*}
P(D)&=\int_0^{D/g}
 \exp\left\{-rt-\int_0^t\lambda(D-gs)\,\de s\right\}\de t
 +\exp\left\{-rD/g-\frac1g\int_0^D\lambda(s)\,\de s\right\}P(0).
\end{align*}
The first term gives discounted profits while waiting for safety to catch up. The second
is the discounted pacing value, weighted by survival until safety.

For $D<0$, the firms race without risk for $-D/(1-g)$ units of time
before pacing. Their normalized payoff is thus 
\[
P(D)=\frac{1-e^{(r-\alpha+\beta)D/(1-g)}}{r-\alpha+\beta}
+e^{(r-\alpha+\beta)D/(1-g)}P(0). 
\]

Let $R(D)$ be the normalized payoff when both firms race forever
from a tie at $D$: 
\[
R(D)=\int_0^\infty
\exp\left\{-(r-\alpha+\beta)t
-\int_0^t\lambda\bigl(D+(1-g)s\bigr)\,\de s\right\}\de t.
\]
Define the \emph{effective hazard} as the constant hazard that gives the payoff from racing starting at the safety boundary
\[
\lambda^{\mathrm{eff}}:=\frac1{R(0)}-(r-\alpha+\beta).
\]
It equals $\lambda_0$ when $\phi=0$ and rises when the risk profile $\phi$
rises pointwisely. It also depends on the safety speed $g$ and the
discounted growth rate $r-\alpha+\beta$. This statistic summarizes
the value of racing from safety.

\FloatBarrier
\begin{definition}[Racing]\label{def:racing}
The \emph{racing profile} $\sigma^R=(\sigma_1^R,\sigma_2^R)$ prescribes
$\sigma_i^R(S):=1$ for each firm $i$ at every state $S$.
\end{definition}

We construct the pacing profile when
$\lambda^{\mathrm{eff}}\geq(1-g)(\alpha-\beta)$, or equivalently when
$R(0)\leq P(0)$. Write
$G:=|A_1-A_2|$ for the \emph{capability gap} between firms. At a strict gap, the leader
has capability $A^L=A$ and the follower has capability $A^F=A-G$. 

Our pacing strategy is as follows: If one firm is strictly behind, the follower races until it catches up (if ever) while the leader chooses its best response $b(D,G)$ given that the follower is racing. If firms are tied, firms choose speeds $a^*(D)$ that prescribe racing whenever the capability frontier is below safety, pacing whenever the capability frontier coincides with safety, and either pausing or racing when the capability frontier is above safety.

\begin{definition}[Pacing]\label{def:pacing}\label{def:profiles}

 If $\lambda^{\mathrm{eff}}\geq(1-g)(\alpha-\beta)$, the \emph{pacing profile}
 $\sigma^P=(\sigma_1^P,\sigma_2^P)$ prescribes, for firm $i$ with
 opponent $j$,
 \[
 \sigma_i^P(S):=\begin{cases}
 a^*(D),&A_i=A_j,\\
 1,&A_i<A_j,\\
 b(D,G),&A_i>A_j.
 \end{cases}
 \]
 \end{definition}

 We now make the construction of $a^*(D)$ and $b(D,G)$ precise.
 If $P(0)=R(0)$, set $\bar D=0$. If $P(0)>R(0)$, let
$\bar D>0$ be the unique solution to $P(\bar D)=R(\bar D)$.
In either case, set  \[
a^*(D):=\begin{cases}
1,&D<0,\\
g,&D=0,\\
0,&0<D\leq\bar D,\\
1,&D>\bar D,
\end{cases}
\]
 and let $\overline{D}>0$ be the unique
distance above safety at which waiting and racing give the same payoff: When $P(0)>R(0)$, waiting gives the higher payoff for
$0<D<\bar D$, and racing gives the higher payoff for $D>\bar D$.
$
P(\bar D)=R(\bar D)
$.\footnote{This cutoff exists and is unique; see \cref{lem:scalar-control} in \cref{sec:proof}.}  If $P(0)\leq R(0)$, set $\bar D=0$.

When one firm is strictly behind, the pacing strategy prescribes that the follower races at max speed until it catches up (if ever). Given this, the leader best responds and chooses an optimal path $(a_t)_t$. Let $V_L(D,G)$ denote the leader's optimal discounted payoff divided by $e^{(\alpha - \beta)A}$, the flow profit each firm would earn if they were both at the current capability frontier. With this normalization, the leader's current flow profit is $e^{\beta G}$ and its Bellman equation is thus:

\begin{equation*}\tag{B}\label{general:eq:leader}
\begin{aligned}
\underbrace{rV_L}_{\text{discounting}}
={}&\underbrace{e^{\beta G}}_{\text{current profit}}
-\underbrace{\lambda(D)V_L}_{\text{disaster loss}}\\[1em]
&+\max_{a\in[0,1]}\Bigl\{
\underbrace{(a-g)\partial_D V_L}_{\substack{\text{change in distance}\\\text{above safety}}}
+\underbrace{(a-1)\partial_G V_L}_{\substack{\text{change in}\\\text{the lead}}}
+\underbrace{(\alpha-\beta)aV_L}_{\substack{\text{growth of}\\\text{frontier flow}}}
\Bigr\}.
\end{aligned}
\end{equation*}
The first term in the continuation value captures how current speeds translate into the distance above safety (since $\dot D = a - g$); the second term captures how current speeds change the lead over the leader's rival (since $\dot G = a - 1$); the last term in the Bellman equation accounts for growth of the
normalizing factor at rate $(\alpha-\beta)a$. Note also that since the pacing strategy pins down speeds after firms are tied i.e., $D = 0$, the boundary condition is $V_L(D,0)=\max\{P(D),R(D)\}$.

We show that facing this problem, the leader can attain its optimal payoffs by racing up to some fixed time, upon which it begins accommodation---pausing until the follower catches up and then both firms choose speeds $a^*$ thereafter. This time is optimally chosen and might be zero, infinite, or positive. 

\begin{figure}[H]
\centering
\includegraphics[width=\textwidth]{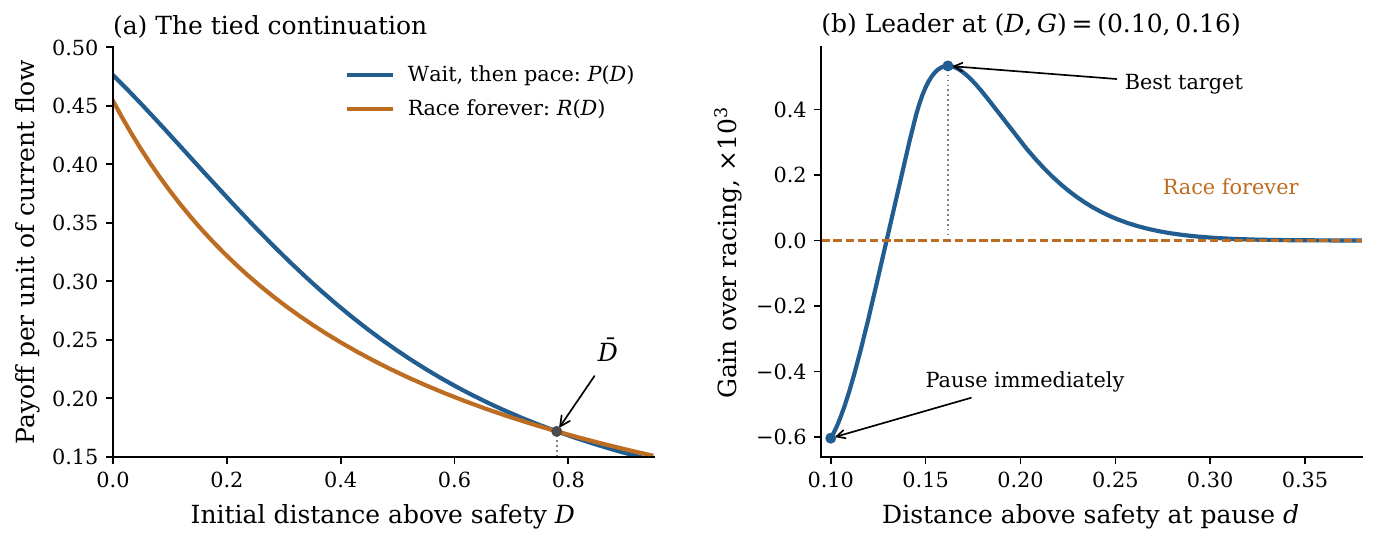}
\begin{singlespacing}
\caption{Continuation values under linear risk}
\label{fig:continuation-values}
\end{singlespacing}
\end{figure}

The leader's continuation is pinned down by the following stopping
formula that gives the optimal value for $0\leq D\leq\bar D$:
\[
\begin{aligned}
&e^{-\beta G}V_L(D,G)=R(D)
+\max\Biggl\{0,\\[.5em]
&\qquad\max_{D\leq d\leq\bar D}
\underbrace{\exp\!\left[-\frac{1}{1-g}
\int_D^d\bigl(r-\alpha+\beta+\lambda(s)\bigr)\,\de s\right]}_{
\substack{\text{discounting, growth, and survival}\\
\text{while racing to }d}}
\underbrace{\bigl[C(d,G)-R(d)\bigr]}_{
\substack{\text{gain from accommodating at }d\\
\text{instead of racing forever}}}
\Biggr\}.
\end{aligned}
\]
The choice $d=D$ is immediate accommodation, the outer zero corresponds to perpetual racing. If a finite $D$  attains the maximum on the RHS, we select the smallest maximizer. \cref{fig:continuation-values}(b) illustrates the payoff from accommodating when the distance above safety is equal to $d > D$.\footnote{Parameters: $\alpha=2$,
$\beta=1$, $r=3$, $g=0.9$, and $\lambda(D) = \kappa D$ where $\kappa=4.8$.}

The leader's best response to pacing is
\[
b(D,G):= \begin{cases}
1, & P(0)=R(0),\\
g, & P(0)>R(0),\ D=0,\ \text{and }
V_L(0,G)=e^{\beta G}C(0,G),\\
0, & P(0)>R(0),\ 0<D\leq\bar D,\ \text{and }
V_L(D,G)=e^{\beta G}C(D,G),\\
1, & \text{otherwise.}
\end{cases}
\]
The equality $V_L(D,G)=e^{\beta G}C(D,G)$ means that accommodating
immediately attains the leader's best payoff. In the first case, the leader is
on the safety threshold ($D=0$) and immediate accommodation is optimal,
so it paces at speed $g$. In the second case, the leader is above safety
but no farther than the tied cutoff ($0<D\leq\bar D$), and immediate
accommodation is optimal, so it pauses. In all other cases, the leader finds it strictly optimal to continue racing, and might do so forever, or accommodate at some future date.

\begin{figure}[H]
\centering
\includegraphics[width=0.9\textwidth]{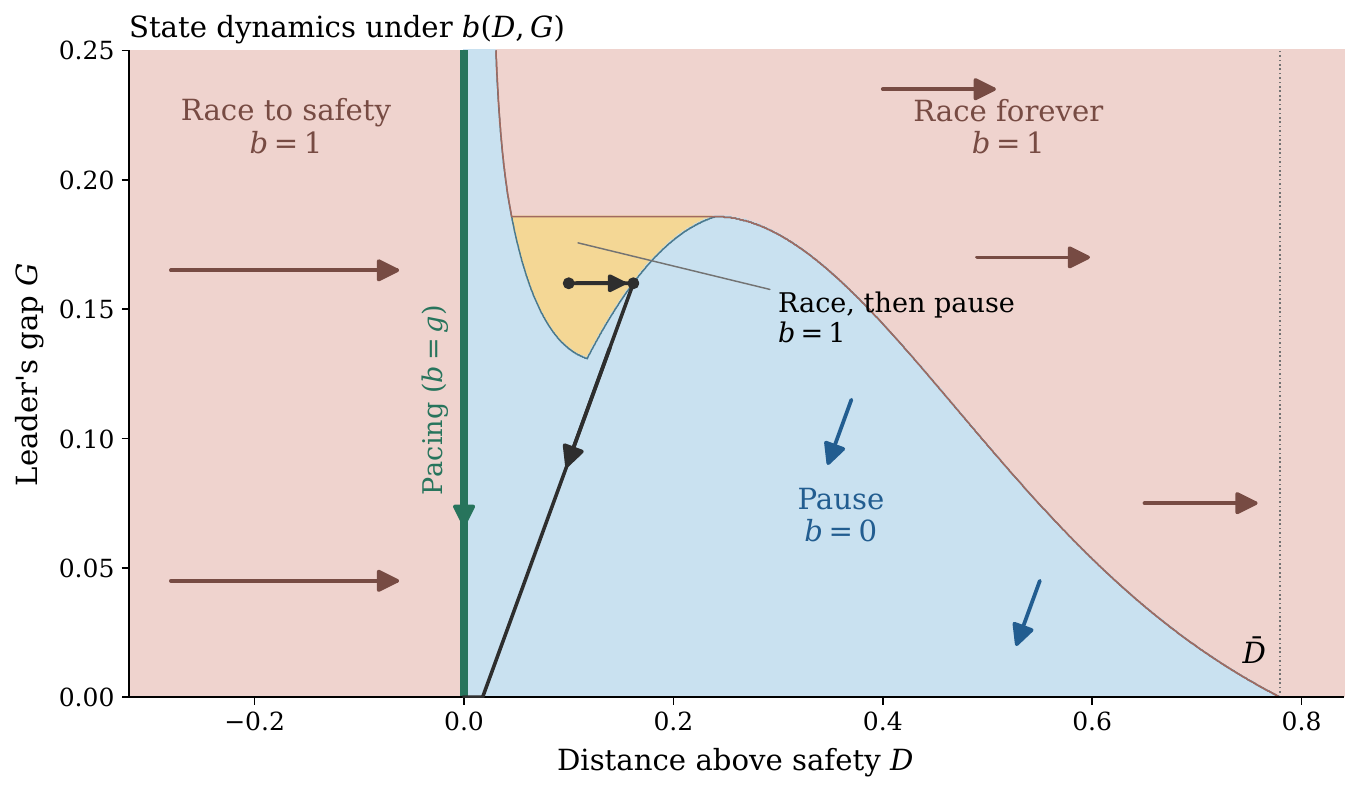}
\caption{Phase diagram induced by $\sigma^P$ under linear risk}
\label{fig:linear-state-strategies}
\end{figure}

\cref{fig:linear-state-strategies} is a phase diagram under the pacing strategy $\sigma^P$ where the $x$-axis is distance above safety $D$ while the $y$-axis is the gap $G$ between leader and follower. Racing moves the state right: $(\dot D,\dot G)=(1-g,0)$.
For $D<0$ it leads to safety; in the red region at $D>0$ it continues
forever. Pausing in blue moves the state down and left:
$(\dot D,\dot G)=(-g,-1)$. On the green line the leader paces and
the gap closes vertically. 

\subsection{Markov Perfect Equilibria}\label{sec:thresholds}

\begin{theorem}[Pacing and racing]\label{thm:main}\label{general:thm:mpe} \phantom{}
\begin{enumerate}[label=(\alph*),leftmargin=2em]
\item \emph{Racing only.} If $\lambda^{\mathrm{eff}}<(1-g)(\alpha-\beta)$, $\sigma^R$
is an \MPE{}, and no \MPE{} outcome paces forever from a tie on safety.
\item \emph{Racing and pacing.} If $(1-g)(\alpha-\beta)\leq\lambda^{\mathrm{eff}}\leq(1-g)\alpha$,
$\sigma^R$ and $\sigma^P$ are both \MPE{}.
\item \emph{Pacing only.} If $\lambda^{\mathrm{eff}}>(1-g)\alpha$, $\sigma^P$
is an \MPE{}, and $\sigma^R$ is not an \MPE{}.
\end{enumerate}
\end{theorem}

\begin{figure}[!ht]
\centering
\includegraphics[width=.8\textwidth]{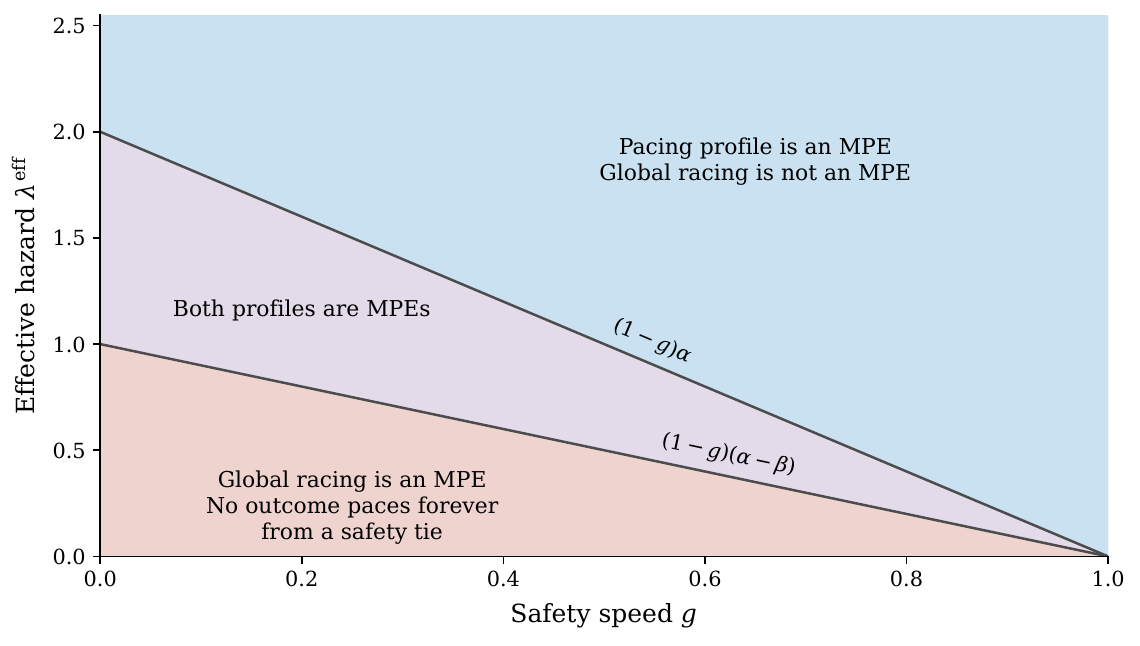}
\caption{Regimes of cooperation}
\label{fig:regimes}
\end{figure}

\cref{general:thm:mpe} is illustrated by \cref{fig:regimes}.  The lower cutoff rules out perpetual pacing in part (a).
Start with both firms at the safety threshold,
$A_{1,0}=A_{2,0}=X_0=A_0$. If firm $i$ races forever, the speed limit implies
$
A_{it}=A_0+t$ and $ A_{jt}\leq A_0+t.$
Firm $i$ therefore determines the frontier, so the distance above safety
is $D_t=(1-g)t$, exactly as under joint racing. Its flow profit satisfies
$
e^{\alpha A_{it}-\beta A_{jt}}
\geq e^{(\alpha-\beta)(A_0+t)}.$ The deviation thus faces the same disaster risk as joint racing and
earns at least as much at every date. It guarantees normalized payoff
$R(0)$, whatever the rival does. Perpetual joint pacing gives $P(0)$.
Hence, when $R(0)>P(0)$, either firm can profitably deviate from pacing
by racing forever. This inequality is equivalent to
$
\lambda^{\mathrm{eff}}<(1-g)(\alpha-\beta).
$
In part (b), perpetual racing and safe pacing are both equilibrium
outcomes from a tie on safety. In the interior of this region, both
firms prefer the pacing outcome, since $P(0)>R(0)$. Yet one firm cannot
obtain that outcome simply by slowing down: its payoff depends on how
its rival responds and what happens when they meet. The pacing profile is an equilibrium when effective hazard satisfies
$
\lambda^{\mathrm{eff}}\geq(1-g)(\alpha-\beta).
$
The racing profile is an equilibrium when
$
\lambda^{\mathrm{eff}}\leq(1-g)\alpha.$ Part (b) is the region where both conditions hold. Risk is high enough
to sustain pacing, but low enough that racing remains optimal against
a rival that always races. Both boundary values belong to this
coexistence region.

\cref{fig:catchup-comparison} illustrates coexistence of these equilibria.\footnote{We use linear risk
with $\alpha=2$, $\beta=1$, $r=3$, $g=1/2$, and
$\kappa\simeq6.4124$. The slope is chosen to give
$\lambda^{\mathrm{eff}}=0.9$, strictly between the two cutoffs.} 
It starts with a small leader on the safety
frontier ($A_{10}=X_0=0.20$ and $A_{20}=0.10$). In both panels, the leader paces while the follower races
until catch-up. Under the pacing profile, both firms then pace, so the
leader's initial slowdown leads to a permanently safe continuation.
Panel (b) shows an unprofitable deviation against the racing profile:
the leader tries pacing until catch-up and then resumes racing, while
the rival races throughout. The slowdown provides only a temporary
safe interval and sacrifices the lead. It gives the leader $0.3396$
per unit of its initial profit flow, below the $0.3448$ from racing
immediately. The leader therefore prefers to race immediately against
this rival; panel (b) shows why the proposed slowdown is rejected. The follower also has an incentive to race, because slowing down  makes immediate accommodation less attractive to the leader.

% Simulation source: figures/cooperating_v3_figures.py
\begin{figure}[!ht]
\centering
\includegraphics[width=\textwidth]{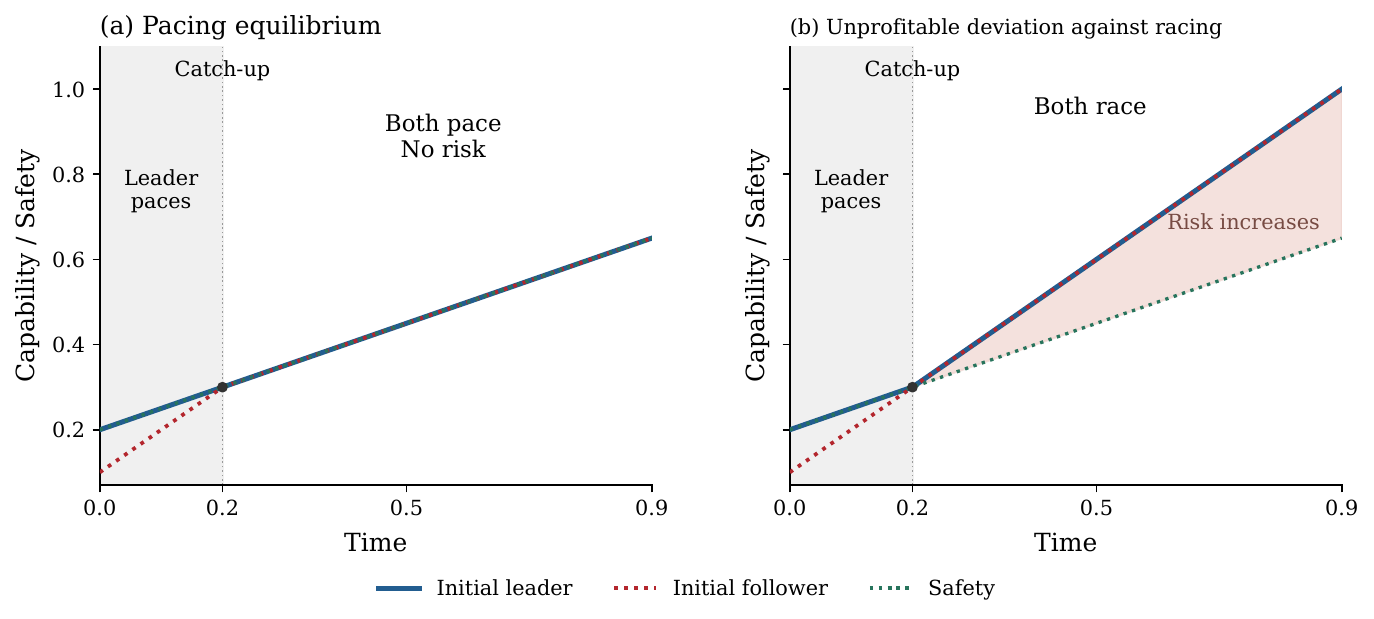}
\caption{Slowdown with different continuations}
\label{fig:catchup-comparison}
\end{figure}
\FloatBarrier

Note that the pacing profile specifies behavior at every state, but it need not
lead to pacing from every initial state.   With
safety far behind, the leader races forever. With safety closer, it
pauses while the follower catches up. This can only be sustained if the safety variable is not too far behind, otherwise waiting would cost too much in forgone growth and exposure to disaster. In this case, the leader would prefer racing and maintaining its lead, even at the cost of disaster risk.

In part (c) of \cref{thm:main}, the pacing profile remains an equilibrium but global
racing does not. Above the upper cutoff, $R(0)<\frac{1}{r+\beta-g\alpha}$. Hence, a leader can gain by unilaterally pacing briefly before
resuming racing against a rival that always races: the brief spell on
safety avoids risk, and the benefit outweighs the loss from letting
the follower narrow the gap and this profitable deviation
rules out the global racing profile. 

% Simulation source: figures/cooperating_v3_figures.py
\begin{figure}[!ht]
\centering
\includegraphics[width=\textwidth]{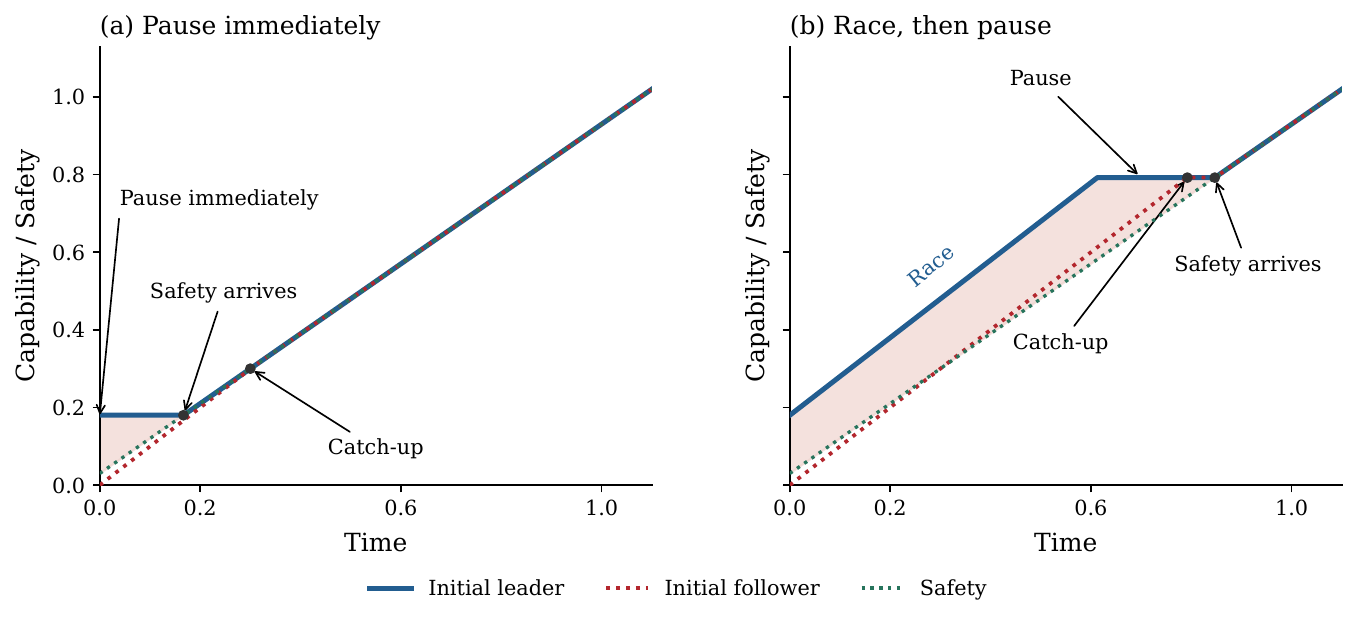}
\caption{Different accommodation timings}
\label{fig:pause-comparison}
\end{figure}
\FloatBarrier

When the leader starts above safety, it might nevertheless prefer a short race before accommodating. \cref{fig:pause-comparison} compares immediate
accommodation with the optimal path from the same state above safety,
using the linear-risk parameters of \cref{fig:linear-state-strategies}.
Here $\lambda^{\mathrm{eff}}\simeq0.2018$ exceeds the upper cutoff $0.2$. Panel (a) illustrates the leader's strategy of pausing immediately until safety catches up, after which the leader paces the frontier while the follower develops at maximum speed. Eventually, the follower catches up, and both firms pace the frontier. Panel (b) illustrates the leader's strategy of first racing before pausing then, when safety catches up, both firms pace the frontier.

\cref{fig:leader-phase-detail} enlarges the part of
\cref{fig:linear-state-strategies} containing the path in
\cref{fig:pause-comparison}(b). Racing moves the state to the right
at a fixed gap: $(\dot D,\dot G)=(1-g,0)$. At the boundary of the
pausing region, the leader stops while the follower keeps racing,
so $(\dot D,\dot G)=(-g,-1)$ and the path turns down and left.
Catch-up occurs before safety arrives. Both firms then pause until$D=0$ and then pace together.

\begin{figure}[!ht]
\centering
\includegraphics[width=.92\textwidth]{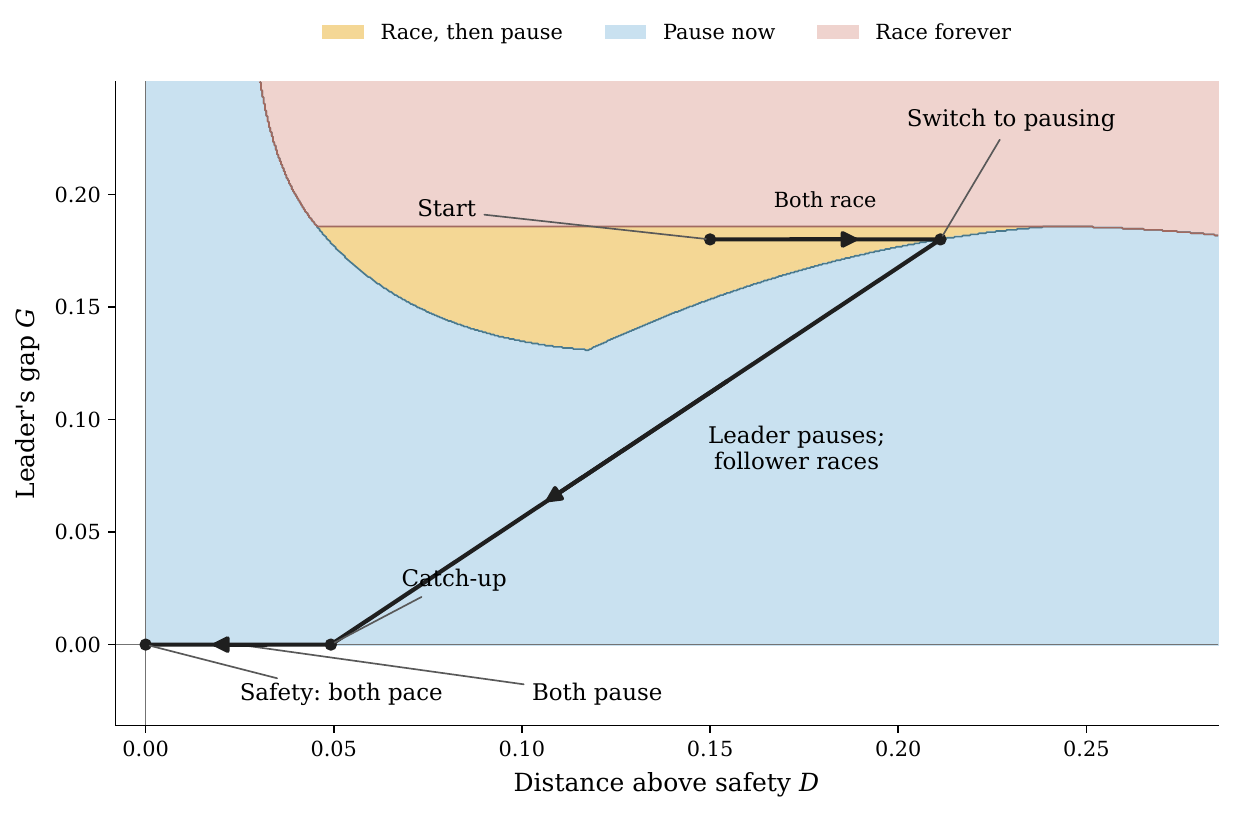}
\caption{The race-then-pause path in state space}
\label{fig:leader-phase-detail}
\end{figure}

\subsection{Comparative statics}
\label{sec:risk-examples}\label{sec:linear-simulations}
 We now flesh out some simple implications of \cref{thm:main}. We first show faster safety progress, higher risk, and greater patience
favor pacing in the equilibrium regions of \cref{thm:main}. We then
use linear risk to examine how these incentives change adjustment
paths and disaster exposure.

\begin{proposition}[Safety progress and patience favor pacing]
\label{prop:comparative-statics}
With other primitives fixed, the ratio \mbox{$\lambda^{\mathrm{eff}}/(1-g)$} is:
\vspace{-1em}
\begin{enumerate}[label=(\alph*),leftmargin=2em]
\item weakly increasing in safety speed $g$;
\item weakly increasing under pointwise increases of the hazard profile $\lambda$;
\item weakly decreasing in the discount rate $r$, and strictly if $\phi$
is nonconstant. 
\end{enumerate}
\vspace{-1em}
Thus faster safety, higher risk, or greater patience weakly expands
the region where $\sigma^P$ is an MPE and weakly shrinks the region
where $\sigma^R$ is an MPE.
\end{proposition}

Faster safety weakly lowers exposure along the racing path, but also reduces
the growth advantage of racing over pacing. Under concave risk, the
second effect dominates for equilibrium incentives. Greater patience
gives more weight to the higher hazards reached later in a race;
this channel is absent under constant risk. The proofs are in \cref{sec:comparative-statics-proofs}.

Under constant risk strictly above safety, effective hazard equals
$\lambda_0$. The two cutoffs are $(1-g)(\alpha-\beta)$ and
$(1-g)\alpha$, both decreasing in $g$ and independent of $r$.
The leader chooses between immediate accommodation and perpetual
racing: any gain from eventually accommodating is largest when
accommodation starts immediately. Under linear risk, $\lambda_0=0$ and $\phi(D)=\kappa D$, the racing
payoff from safety is
\[
R(0)=\int_0^\infty
e^{-(r-\alpha+\beta)t-\kappa(1-g)t^2/2}\,\de t.
\]
Let the lower and upper cutoff slopes be $\underline\kappa$ and
$\overline\kappa$, defined respectively by $R(0)=P(0)$ and
$R(0)=1/(r+\beta-g\alpha)$. Pacing exists for
$\kappa\geq\underline\kappa$, global racing exists for
$\kappa\leq\overline\kappa$, and both exist between the cutoffs,
including their endpoints.

\subsection{Subgame Perfect Equilibria}\label{sec:spe-outcomes}
We now allow equilibrium strategies to depend on the full history of play, and study  how competitive dynamics shape risk across all subgame perfect equilibria.

\begin{definition}[Subgame-perfect equilibrium]\label{spe:def:spe}
A pure profile $\sigma$ is a {subgame-perfect equilibrium} (SPE) if,
after every feasible history $h_t$, for each firm $i$ and
every alternative history-dependent strategy $\sigma_i'$,
$
V_i((\sigma_i',\sigma_j),h_t)\leq V_i(\sigma,h_t).$ 
An \emph{SPE outcome} is an admissible continuation that attains both
assigned payoffs on the same path.
\end{definition}

We establish that when risk is low and the future is discounted sufficiently, firms perpetually race across all SPE. We also bound  the probability of disaster up to a fixed time. To show this, we first construct a lower-bound $\underline{V}_{\infty}$ on firms' payoffs across all SPE. We do so recursively, as illustrated in
\cref{fig:recursive-floor}. At the rightmost safety tie, each firm can guarantee its racing payoff
$\underline{V}_0=R(0)$. Thus, at the middle tie, firm 2 (red) can pause which gives 
firm 1 (blue) the option to wait before pacing to the terminal guarantee. The logic here is that because disaster ends both firms’ profits, firm 1’s incentive to avoid disaster also protects firm 2. This gives firm 2 the improved payoff guarantee \(\underline V_1\). Now working backward, at the initial tie, firm $1$ can in turn pause and---with the roles reversed---use firm 2's continuation
guarantee $\underline{V}_1$ to obtain $\underline{V}_2$.
Repeating the argument gives guarantees converging to
$\underline{V}_\infty$.
\begin{figure}[H]
\centering
\includegraphics[width=0.9\textwidth]{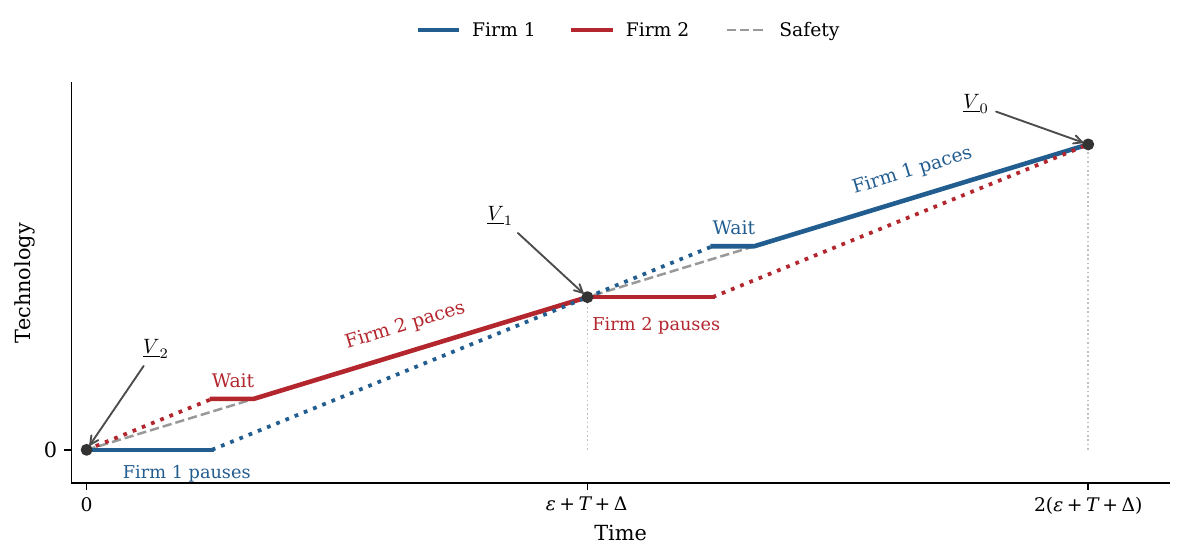}
\caption{Constructing the recursive payoff guarantee}
\label{fig:recursive-floor}
\end{figure}
Next, we construct an upper bound $\overline V(t,z)$ on average discounted profits across all SPE when cumulative disaster exposure reaches \(z\) by date \(t\). To do so, we first observe that cumulative disaster by time $t$ is upper-bounded by 
$
\Lambda_t^R:=\int_0^t\lambda\big( (1-g)s \big)\,\de s$,
since $(1-g)s \geq D_s$. Next notice that $F(t):=\max\left\{e^{(\alpha-\beta)t},
\frac{e^{\alpha t}+e^{-\beta t}}{2}\right\}$ is the largest average flow profit by date $t$.\footnote{The first term is if it maximizes total profits for both firms to be tied, and the second is if one firm races while the other's capability is at $0$.} Thus define 
\[
\overline V(t,z)
:=\underbrace{\int_0^t e^{-rs}F(s)
\exp\!\left[-\max\{0,z-\Lambda_t^R+\Lambda_s^R\}\right]\,\mathrm ds}_{\substack{\text{Profits through }t\\\text{with exposure postponed as long as possible}}}
+\underbrace{e^{-z}\int_t^\infty e^{-rs}F(s)\,\mathrm ds}_{\substack{\text{Profits after }t\\\text{with no further exposure}}}.
\]
which is decreasing in $z$. Thus, the exposure ceiling
\[
\overline{\Lambda}(t):=
\max\left\{
z\in[0,\Lambda_t^R]:
\overline V(t,z)\geq\underline V_\infty
\right\}
\]
gives the largest cumulative exposure consistent with our SPE payoff guarantees.

\begin{theorem}[Racing and risk bounds in every SPE]\label{thm:spe} \phantom{}
\begin{enumerate}[label=(\alph*),leftmargin=2em]
\item \emph{Low risk forces racing.} If
$\lambda^{\mathrm{eff}}<(1-g)(\alpha-\beta)$ and $
r+\lambda^{\mathrm{eff}}>\frac{\alpha^2+\beta^2}{\alpha-\beta},$  
 every SPE outcome starting from a safety tie is perpetual joint racing.
\item \emph{Upper bound on risk.}  If $r>\alpha$ and
$g>1/2$ and firms start from a tie on the safety threshold, then every SPE outcome satisfies
$
\Pr(\tau_D\leq t)\leq1-e^{-\overline\Lambda(t)}$ for every finite date $t\geq0$. Also, every SPE gives each firm normalized payoff at least the amount 
$\underline{V}_\infty$.
\end{enumerate}
\end{theorem}

Part (a) of \cref{thm:spe} shows that firms cannot cooperating on pacing the frontier when risks are sufficiently low. The underlying idea is that each firm can guarantee its racing payoff $R(0)$ against any response, so an SPE must deliver at least $2R(0)$ in total profits. But under part (a)'s assumptions, low risk makes delay costly, while sufficient discounting
limits the future gains from creating a capability lead. Hence, the total payoff from racing of $2R(0)$ attains the upper bound while every other outcome gives strictly less and so cannot be an SPE.

The next result shows that history-dependent cooperation can give both firms strictly more than the pacing payoff in the region where racing and pacing coexist as MPEs. 

\begin{proposition}[An SPE with higher payoffs than pacing]\label{prop:high SPE}
Suppose  $\lambda(D)=\lambda$ for every $D>0$, so that
$\lambda^{\mathrm{eff}}=\lambda
$. and that $(1-g)(\alpha-\beta)\leq\lambda\leq(1-g)\alpha$ and $r+\lambda<\alpha.$ Starting from $A_{10}=A_{20}=X_0=0$, there is an SPE outcome
giving both firms strictly more than the pacing-MPE payoff
$P=1/[r-g(\alpha-\beta)]$.
\end{proposition}
The key here is that when $r+\lambda<\alpha$  a firm’s discounted flow rises during its own risky development turn, making sufficiently delayed reciprocal development valuable enough to support.

Finally, we analyze the extent to which the threat of future racing can be used to discipline pacing. For this result we  generalize firms' profit functions to any  positive, continuously
differentiable functions $\pi_i(A_i,A_j)$ such that for each firm $i$ $R_i(S)<\infty$ for every state $S$ and
$P_i(x)<\infty$ for every $x\in\mathbb{R}$,
and at every capability pair $A_i, A_j$,
$\dfrac{\partial\log\pi_i(A_i,A_j)}{\partial A_i}\geq\alpha>0$ and $
\dfrac{\partial\pi_i(A_i,A_j)}{\partial A_j}\leq0.$ 
Here $\alpha$ is a common lower bound on the proportional return to
own development. We also assume a finite
hazard ceiling $\sup_{D>0}\lambda(D) =: \bar \lambda \leq (1-g) \alpha$.\footnote{The argument holds for more general increasing hazards that satisfy the same bound.}

From a state $S=(A_1,A_2,X)$, with $D=\max\{A_1,A_2\}-X$, define
the normalized mutual-racing payoff $R_i(S)$ and the normalized
joint-pacing payoff $P_i(x)$ from a safety tie at capability $x$ by
\begin{align*}
R_i(S)&:=\int_0^\infty
e^{-rs-\int_0^s\lambda(D+(1-g)t)\,\de t}
\frac{\pi_i(A_i+s,A_j+s)}{\pi_i(A_i,A_j)}\,\de s,\\
P_i(x)&:=\int_0^\infty e^{-rs}
\frac{\pi_i(x+gs,x+gs)}{\pi_i(x,x)}\,\de s.
\end{align*}
Under exponential profits, the safety-tie values
reduce to $R_i(x,x,x)=R(0)$ and $P_i(x)=P(0)$.

\begin{proposition}[Pacing is sustainable if and only if it beats racing]
\label{prop:pacing-iff}
Permanent joint pacing from $A_{10}=A_{20}=X_0=A_0$ is an
SPE outcome if and only if for each firm $i$ and each $x \geq A_0$, 
\begin{equation*}
P_i(x)\geq R_i(x,x,x)
\end{equation*}
When this holds, a grim trigger implements pacing:
both firms pace while the history follows the prescribed pacing
path, and race forever after any departure.
\end{proposition}

The idea driving \cref{prop:pacing-iff} is that when  $\lambda \leq (1-g) \alpha$, mutual racing is a credible continuation
that attains both firms' lowest equilibrium payoffs at every state. With exponential profits, the payoff comparison reduces to $P(0)\geq R(0)$, or equivalently
$\lambda^{\mathrm{eff}}\geq(1-g)(\alpha-\beta)$ which is exactly the condition for permanent pacing in \cref{thm:main} which implemented it in Markov strategies.

\subsection{Monitoring lags: internal vs deployed models}
We show that low effective hazard rules out pacing, intermediate
risk can sustain pacing when deployment is sufficiently prompt, and
high effective hazard sustains pacing at every deployment lag.

The two firms develop capability at speeds in $[0,1]$.
An internal model becomes public and is deployed after a fixed lag
$L>0$. We assume the
initial history commonly known history is that both firms developed at speed $g$: $A_{is}=X_s+gL, s\in[-L,0], i=1,2,$so $A_0=X_0+gL$.
At date $t$, deployed capabilities determine profits and risk:
\[
\pi_{it}:=e^{\alpha A_{i,t-L}-\beta A_{j,t-L}}
\quad\text{and}\quad
\lambda_t=\lambda(A_{t-L}-X_t).
\]
We retain the hazard profile in \eqref{eq:risk}
and assume a finite ceiling $\bar\lambda:=\sup_{D>0}\lambda(D)$.
The initial histories place deployed capabilities at the safety threshold during
$[0,L]$, so profits over that interval are fixed and hazard is zero.

Firm $i$ observes its own development through date $t$ and
its rival's capability through $\max\{0,t-L\}$; both firms observe
survival and the safety threshold $X_t=X_0+gt$.

\paragraph{Strategies and equilibrium.}
Firm $i$'s private history $h_{it}$ records its own development,
its observations of its rival's capabilities, and survival through $t$.
A private history is feasible if there exists a pair of absolutely
continuous capability histories extending the initial histories, with
speeds in $[0,1]$ almost everywhere, that reproduces the firm's own
development and its observed rival deployments.

Let $H_{it}$ denote firm $i$'s
private history. Firm $i$'s information is described by
the filtration $
\mathcal F_{it}:=\sigma(H_{is}:0\le s\le t).
$ 
A pure strategy $\sigma_i$ specifies a speed
process $(a_{it})_{t\ge0}$ that is jointly measurable, $[0,1]$-valued, and adapted to $(\mathcal F_{it})_{t\ge0}$. As in our baseline model, strategies map to paths via regularization. Write $h_t$ for the full
history on $[-L,t]$ and let $\mathcal P(\sigma;h_0)$ denote the admissible paths from the common initial
history $h_0$, and each firm evaluates a strategy
according to its best admissible path.

\begin{definition}[Pure Strategy  Nash Equilibrium]
\label{def:deployment-nash}
A pure strategy profile $\sigma$ is a \emph{Nash equilibrium} if, for
each firm $i$ and every pure strategy $\sigma_i'$,
$
V_i(\sigma)\ge V_i((\sigma_i',\sigma_j)).
$ A \emph{Nash equilibrium outcome} is an admissible path attaining
both firms' assigned payoffs on the same path, with play stopped at
disaster.
\end{definition}

\begin{definition}[Joint pacing under delayed deployment]
\label{def:deployment-pacing}
An equilibrium outcome sustains \emph{joint pacing} if, conditional on survival, $A_{1,t-L}=A_{2,t-L}=X_t$ at every date $t\ge0$.
\end{definition}

\paragraph{Deployment lags that sustain pacing.}

If a firm races from time $0$, the capability of its deployed technology at time $L + s$ exceeds safety by $(1-g)s$ so 
$\Lambda_s^R:=\int_0^s\lambda((1-g)t)\,\de t$
is the associated cumulative disaster exposure through time $L+s$.

Let $B(L)$ be the normalized
payoff from  racing against a rival that
paces for $L$ units of time and then races forever:
\[
B(L)
=\underbrace{\int_0^L e^{-(r-\alpha+\beta g)s-\Lambda_s^R}\,\de s}
_{\text{rival paces before detection}}
\quad+\underbrace{e^{\beta(1-g)L}\int_L^\infty
e^{-(r-\alpha+\beta)s-\Lambda_s^R}\,\de s}
_{\text{rival races after detection}}.
\]
As before, pacing gives value
$P(0)=1/[r-g(\alpha-\beta)]$.\footnote{Starting from total payoffs, we first subtract profits over $[0,L]$ (that cannot be changed by speed choices from $0$) and then divide by $e^{-rL}$ since current speed choices only affect profits after $L$ units of time, and divide by $e^{(\alpha - \beta)A_0}$ which is the profit scale from initial capabilities. Hence, the pacing payoff under this normalization is $P(0)$.}

\Needspace{18\baselineskip}
\begin{proposition}[Risk and delay determine sustainable pacing]
\label{prop:deployment-lag}
For every lag $L>0$, joint racing is a pure Nash equilibrium outcome.
Pacing is sustainable under the following conditions: 
\begin{enumerate}[label=(\alph*),leftmargin=2em]
\item If $\lambda^{\mathrm{eff}}<(1-g)(\alpha-\beta)$, then no pure
Nash equilibrium sustains joint pacing for any $L>0$.

\item If $\lambda^{\mathrm{eff}}\ge(1-g)(\alpha-\beta)$ and
$\bar\lambda<(1-g)\alpha$, then joint pacing is a pure Nash
equilibrium outcome if and only if $L\le L^*$, where $L^*\ge0$ is
the unique finite solution to $B(L^*)=P(0)$.

\item If $\lambda^{\mathrm{eff}}\ge(1-g)\alpha$, then joint pacing
is a pure Nash equilibrium outcome for every $L>0$.
\end{enumerate}
\end{proposition}

The strategy that sustains pacing in part (b) is as follows: a firm develops at speed \(g\) if it has always paced and all its observations of its rival are consistent with pacing. Otherwise, it prescribes speed \(1\).\footnote{The strategy for (c) is different and simply prescribes speed $g$ after every history. But we will focus our explanation on (b).} Hence, the prospect that (belated) discovery that one has deviated might start an AI race disciplines incentives to pace along the safety threshold. This is also why pacing becomes more difficult to sustain as the lag grows: firms can accelerate to speed $1 > g$ undetected for $L$ units of time and, in so doing, build up a lead. By the time this deviation is detected, its rival might race in response but it is already too late: the deviating firm would have already built up an enduring lead of $L \cdot (1-g)$, and enjoys elevated profits until disaster strikes.

\Cref{prop:deployment-lag} is stated for Nash equilibrium, making
the impossibility of sustaining pacing in part (a) stronger.
The construction in part (b) also makes racing optimal after
a detected departure, independently of beliefs about hidden rival
development. Fix any feasible private history at which racing is
prescribed, any consistent hidden rival history, and any feasible
continuation of the rival's development. Holding that rival path
fixed, sacrificing one unit of capability costs $\alpha$ units of
log profit but saves at most $\bar\lambda/(1-g)$ units of cumulative
hazard. Since $\bar\lambda/(1-g)<\alpha$, racing yields at least as
much survival-adjusted continuation payoff as any slower 
development path.

\cref{prop:deployment-lag} implies that firms can benefit from committing to revealing internal progress before deployment. This might reflect, for instance, the use of a third-party `embedded evaluator' like Model Evaluation and Threat Research (METR) that tracks the state of internal R\&D. Doing so in effect shortens or eliminates the the lag between internal development and public observation, which  allows pacing to be sustained for a wider set of parameters: when the capabilities of  models that have not yet been deployed are  public, makes deviating from pacing to racing is less appealing since doing so would be quickly met with racing.

\setlength{\bibsep}{0pt}
\bibliographystyle{abbrvnat}
\bibliography{ref}

\appendix
\crefalias{section}{appendix}
\crefalias{subsection}{appendix}
\normalsize
\section{Proofs}\label{sec:proof}

\subsection{Histories and verification}\label{sec:proof-preliminaries}
\label{sec:history-paths}

\paragraph{Admissible continuations after histories.}
We use the histories and pure strategies of \cref{sec:strategies},
writing the prescribed speed pair at history $h_t$ as $\sigma(h_t)$.
To extend regularization to histories, we keep each history constant after
its terminal date\footnote{That is, $h_t(t') = S_{t'}$ for $t' \leq t$ and $h_t(t') = S_{t'}$ for $t' \geq t$.} and measure distance by
\begin{equation*}
d(h_t,\widetilde h_{t'})
:=|t-t'|+\sup_{s\geq0}
\bigl\|h_t(s\wedge t)-\widetilde h_{t'}(s\wedge t')\bigr\|.
\end{equation*}
where $\|\cdot \|$ is the Euclidian norm on $\mathbb{R}^3$.  The supremum is finite because the two stopped histories are constant
beyond their terminal dates.  Define the joint Krasovskii set by
\begin{equation*}\tag{HK}\label{spe:eq:regularization}
K[\sigma](h_t)
:=\bigcap_{\varepsilon>0}\overline{\operatorname{co}}
\bigl\{\sigma(\widetilde h_{t'}):
 d(h_t,\widetilde h_{t'})<\varepsilon\bigr\}.
\end{equation*}

An admissible continuation after $h_t$ extends that fixed history by an
absolutely continuous state path satisfying, almost everywhere for $s>t$,
\begin{equation*}\tag{HA}\label{spe:eq:path}
(\dot A_{1s},\dot A_{2s})\in K[\sigma](h_s)
\quad\text{and}\quad
\dot X_s=g.
\end{equation*}
Write $\mathcal P(\sigma;h_t)$ for these continuations, where $h_s$
includes the original history and the extension through date $s$. 

\begin{lemma}[Regularity on histories]\label{lem:nice}
For every pure strategy profile $\sigma$, the correspondence
$K[\sigma]$ defined by \cref{spe:eq:regularization} has nonempty compact
convex values, is globally bounded, and is upper semicontinuous in
the stopped-history metric.
\end{lemma}
\begin{proof}
For $\varepsilon>0$, let $K_\varepsilon[\sigma](h)$ be the closed
convex hull of the prescriptions at feasible histories within
$\varepsilon$ of $h$. It is a nonempty compact convex subset of
$[0,1]^2$. As $\varepsilon$ decreases, these sets are nested, and
$K[\sigma](h)=\bigcap_{n\geq1}K_{1/n}[\sigma](h)$. Thus
$K[\sigma](h)$ is nonempty, compact, and convex. The inclusion
$K[\sigma](h)\subseteq[0,1]^2$ gives global boundedness.

It remains to prove upper semicontinuity. Suppose $h^n\to h$ and
$v^n\to v$, with $v^n\in K[\sigma](h^n)$. For any $\varepsilon>0$,
eventually $d(h^n,h)<\varepsilon/2$, and hence
$K[\sigma](h^n)\subseteq K_\varepsilon[\sigma](h)$. Closedness then
gives $v\in K_\varepsilon[\sigma](h)$. Since this holds for every
$\varepsilon>0$, $v\in K[\sigma](h)$, so the correspondence has a
closed graph. A closed-graph correspondence whose values lie in the
common compact set $[0,1]^2$ is upper semicontinuous.
\end{proof}

\begin{lemma}[Paths after histories]\label{spe:lem:paths}
Fix a pure strategy $\sigma$ and a feasible history $h_t$.
\begin{enumerate}[label=(\roman*),leftmargin=2em,itemsep=2pt,parsep=0pt]
\item \emph{Existence.} There is an admissible continuation after $h_t$:
$
\mathcal P(\sigma;h_t)\neq\varnothing.
$\item \emph{Restriction and concatenation.} If
$S\in\mathcal P(\sigma;h_t)$, then, for every $s\geq t$,
$
S|_{[s,\infty)}\in\mathcal P(\sigma;h_s),$
where $h_s$ is the full history generated by $h_t$ and $S$.
If $\widetilde S\in\mathcal P(\sigma;h_s)$, then the joined path
\[
\widehat S_v:=
\begin{cases}
S_v,&t\leq v\leq s,\\
\widetilde S_v,&v>s
\end{cases}
\]
satisfies $\widehat S\in\mathcal P(\sigma;h_t)$.
\item \emph{Consistency with Markov paths.} If $\sigma$ is Markov,
then history regularization equals state regularization:
$
K[\sigma](h_s)=K[\sigma](S_s)$, and hence \cref{spe:eq:path} coincides with \cref{eq:generalized-path}.

\item \emph{Deviations against a Markov opponent.} If $\sigma_j$ is
Markov, then, for every history-dependent deviation $\sigma_i'$,
\[
S\in\mathcal P((\sigma_i',\sigma_j);h_t)
\quad\Longrightarrow\quad
\begin{cases}
\dot A_{is}=a_i(s),\\
\dot A_{js}\in K[\sigma_j](S_s),\\
\dot X_s=g,
\end{cases}
\quad\text{for almost every }s>t,
\]
for some measurable speed path $a_i:[t,\infty)\to[0,1]$.
These are the conditions in \cref{eq:generalized-deviation}.
\end{enumerate}
\end{lemma}

\begin{proof}[Proof of \cref{spe:lem:paths}]

\emph{(i) Existence.}
We follow the approximation and compactness argument of
\citet[proof of Theorem 1.1]{Haddad1984}. By \cref{lem:nice}, the history correspondence
$h \mapsto K[\sigma](h)$ has nonempty compact convex values,
is globally bounded, and is upper semicontinuous.

We construct a continuation directly fixing the history $h_t$. For each positive integer $n$, partition the future at
dates $t+k/n$, $k=0,1,\ldots$. Denote the constructed state path by $S^n$
and its development speed pair by $a^n$. Start by setting
$S^n_{t'}=S_{t'}$ for $0\leq t'\leq t$. Once the path has been
constructed through $t+k/n$, let $h^n_{t+k/n}$ be its full history
and choose $
a^n(t+k/n)\in K[\sigma](h^n_{t+k/n}).$
Extend the state path to the next partition date by
\[
S^n_{t'}=S^n_{t+k/n}
+\left(t'-t-\frac{k}{n}\right)
\begin{pmatrix}
a^n_1(t+k/n)\\
a^n_2(t+k/n)\\
g
\end{pmatrix},
\qquad t'\in\left[t+\frac{k}{n},t+\frac{k+1}{n}\right].
\]
Set $a^n(t')=a^n(t+k/n)$ on the corresponding half-open interval
$[t+k/n,t+(k+1)/n)$. Repeating this construction for
$k=0,1,\ldots$ defines the entire path and its piecewise constant
development speeds. Each path is absolutely continuous on every finite interval: it agrees
with the feasible history before $t$ and has finitely many linear
pieces on each bounded interval after $t$. Since the chosen development
speeds lie in $[0,1]^2$, every constructed history is feasible and
\[
\|\dot S^n_{t'}\|^2
=(\dot A^n_{1t'})^2+(\dot A^n_{2t'})^2+g^2
\leq 2+g^2
\quad\text{for almost every }t'\geq0.
\]
Here the same bound holds before $t$ by feasibility of the fixed past.
Absolute continuity therefore gives, for every $0\leq a<b$,
\[
\|S^n_b-S^n_a\|
\leq\int_a^b\|\dot S^n_{t'}\|\,\de t'
\leq\sqrt{2+g^2}\,(b-a).
\]
Thus the paths share a Lipschitz constant independent of $n$.

Fix a finite horizon $T>t$. Since all paths start at the same state
$S_0$, the preceding inequality implies $
\sup_{0\leq t'\leq T}\|S^n_{t'}\|
\leq\|S_0\|+\sqrt{2+g^2}\,T,$  so the restrictions to $[0,T]$ are uniformly bounded and
equicontinuous. By Arzel\`a--Ascoli, they are relatively compact in
$\mathcal C([0,T],\mathbb R^3)$ with the uniform norm.
Extract a subsequence converging uniformly on \([0,t+1]\), then from it a further subsequence converging uniformly on \([0,t+2]\), and continue in this way. The diagonal subsequence converges uniformly on $[0,t+m]$ for every positive integer $m$, and hence locally uniformly on $[0,\infty)$; these limits therefore define a single continuous path $S$ on $[0,\infty)$.Relabel this subsequence with $n$.
The limit inherits the common Lipschitz bound and agrees with $h_t$.
Passing to the limit in the capability and safety increments gives,
for $t\leq a<b$, $
0\leq A_{ib}-A_{ia}\leq b-a$ and $X_b-X_a=g(b-a),$ so  the limit is feasible. It remains to show that the limiting path’s development speeds belong to $K[\sigma](h_{t'})$ at almost every date, where $h_{t'}$ is the full history generated by that path.

Let $\eta_n(t')$ be the partition date immediately preceding $t'$ (inclusive of $t'$ whenever it is a partition date). For every finite
$T>t$ and $t'\in[t,T]$, the stopped histories satisfy
\[
d(h^n_{\eta_n(t')},h_{t'})
\leq \sup_{0\leq t''\leq T}\|S^n_{t''}-S_{t''}\|
     +\frac{1+\sqrt{2+g^2}}{n}
\longrightarrow0.
\]
By construction, $a^n(t')\in K[\sigma](h^n_{\eta_n(t')})$.
Fix $q\in\mathbb Q^2$ and a date $t'$. Choose a subsequence along
which $q\cdot a^n(t')$ converges to its limsup. Compactness of
$[0,1]^2$ gives a further subsequence along which $a^n(t')$ converges
to some vector $a$. Since
$h^n_{\eta_n(t')}\to h_{t'}$, the closed-graph property implies that
$a\in K[\sigma](h_{t'})$. It follows that
$\limsup_{n\to\infty}q\cdot a^n(t')
\leq\max_{v\in K[\sigma](h_{t'})}q\cdot v$.

To verify measurability of the right-hand side, define the support
function $H_q(h):=\max_{v\in K[\sigma](h)}q\cdot v$. This function is
upper semicontinuous on the space of stopped histories. The common
Lipschitz bound gives
$d(h_s,h_{s'})\leq(1+\sqrt{2+g^2})|s-s'|$, so $s\mapsto h_s$ is
continuous in the stopped-history metric. Therefore
$s\mapsto H_q(h_s)$ is upper semicontinuous and hence measurable.
The functions $q\cdot a^n(t')$ and $H_q(h_{t'})$ are uniformly
bounded. Applying
reverse Fatou on any interval $[a,b]\subseteq[t,T]$ and using uniform
convergence of the capability paths gives
\[
q\cdot
\begin{pmatrix}A_{1b}-A_{1a}\\ A_{2b}-A_{2a}\end{pmatrix}
=\lim_{n\to\infty}\int_a^b q\cdot a^n(t')\,\de t'
\leq\int_a^b\max_{v\in K[\sigma](h_{t'})}q\cdot v\,\de t'.
\]
Since this integral inequality holds on every subinterval, Lebesgue differentiation gives the corresponding inequality between the integrands almost everywhere: $q\cdot(\dot A_{1t'},\dot A_{2t'})
\leq\max_{v\in K[\sigma](h_{t'})}q\cdot $ for almost every $t'\in[t,T].$  
Since $\mathbb Q^2$ is countable, these inequalities hold simultaneously outside a null set. At each remaining date, both sides are continuous in $q$, so density extends the inequality to every $q\in\mathbb R^2$. Thus, in every direction $q$,  the speed pair projects no farther than the allowed set does. By \cref{lem:nice}, \(K[\sigma](h_{t'})\) is closed and convex. The separating-hyperplane theorem therefore implies that the realized speed pair belongs to this set: 
$
(\dot A_{1t'},\dot A_{2t'})\in K[\sigma](h_{t'})$ for almost every $t'\in[t,T].$ Taking a countable sequence of horizons increasing to infinity proves
\cref{spe:eq:path} on the entire continuation. This establishes
existence.

\emph{(ii) Restriction and concatenation.}
Fix $S\in\mathcal P(\sigma;h_t)$ and a later date $t'\geq t$.
By admissibility, there is a null set $N\subset(t,\infty)$ such that
$
(\dot A_{1t''},\dot A_{2t''})\in K[\sigma](h_{t''}$ and $\quad \dot X_{t''}=g$ for every $t''>t$ $N.$ 
For restriction, retain $h_{t'}$ as the fixed past. Joining this past
to $S|_{[t',\infty)}$ reproduces the original full history $h_{t''}$
at every date $t''\geq t'$. The displayed conditions therefore hold
on $(t',\infty)$ outside the null set $N\cap(t',\infty)$.
Absolute continuity is preserved under restriction, so
$S|_{[t',\infty)}\in\mathcal P(\sigma;h_{t'})$.

For concatenation, take
$\widetilde S\in\mathcal P(\sigma;h_{t'})$ and form $\widehat S$
as in the statement, with capability coordinates $\widehat A_1,
\widehat A_2$ and safety coordinate $\widehat X$.
Let $\widetilde h_{t''}$ denote the full history
of the appended continuation, including the fixed past $h_{t'}$,
and let $\widehat h_{t''}$ denote the full history of the joined path.
These histories satisfy
\[
\widehat h_{t''}=
\begin{cases}
h_{t''},&t\leq t''\leq t',\\
\widetilde h_{t''},&t''>t'.
\end{cases}
\]
The pieces agree at the joining date, since
$\widetilde S_{t'}=S_{t'}$. Thus $\widehat S$ is continuous and,
being absolutely continuous on each side of $t'$, is absolutely
continuous on every finite interval. Its derivative equals that of
$S$ before $t'$ and that of $\widetilde S$ after $t'$, wherever those
derivatives exist.

Let $\widetilde N\subset(t',\infty)$ be a null set outside which
the appended continuation satisfies \cref{spe:eq:path}.
The derivative and history identities give
$
(\dot{\widehat A}_{1t''},\dot{\widehat A}_{2t''})
\in K[\sigma](\widehat h_{t''})$ and $ \dot{\widehat X}_{t''}=g$ 
for every $t''>t$ outside $N\cup\widetilde N\cup\{t'\}$.
This union is a null set, and the joined path extends $h_t$.
Hence $\widehat S\in\mathcal P(\sigma;h_t)$.

\emph{(iii) Consistency with Markov paths.}
Fix a history $h_{t'}$ ending at $S_{t'}$. For any nearby history
$\widetilde h_{t''}$, its terminal state satisfies
$
\|\widetilde S_{t''}-S_{t'}\|
\leq d(h_{t'},\widetilde h_{t''}).$ 
If the profile is Markov, its prescription at that history is its
prescription at $\widetilde S_{t''}$. Taking closed convex hulls and then intersecting over neighborhood sizes gives
$K[\sigma](h_{t'})\subseteq K[\sigma](S_{t'})$.
Conversely, for any state $S'$ near $S_{t'}$, translate every state of
$h_{t'}$ by the constant vector $S'-S_{t'}$. The translated history is
feasible, ends at $S'$, and has stopped-history distance
$\|S'-S_{t'}\|$ from $h_{t'}$. Such translations are allowed because
\cref{spe:eq:regularization} does not fix the initial state of nearby
histories. Every prescription in a state neighborhood is therefore
also obtained in the corresponding history neighborhood. Taking
closed convex hulls and intersections gives the reverse inclusion,
and hence $K[\sigma](h_{t'})=K[\sigma](S_{t'})$.

\emph{(iv) Deviations against a Markov opponent.}
Fix a history-dependent deviation $\sigma_i'$ and a Markov opponent
$\sigma_j$. The endpoint bound in part (iii) implies that the
opponent's coordinate of every prescription in an $\varepsilon$
history neighborhood belongs to its prescriptions in the
$\varepsilon$ state neighborhood. Projection preserves convex
combinations and limits. Taking closed convex hulls and then
intersecting over $\varepsilon>0$ therefore gives
$
\operatorname{proj}_j K[(\sigma_i',\sigma_j)](h_{t'})
\subseteq K[\sigma_j](S_{t'}).
$
Along any admissible deviation path, set $a_i(t')=\dot A_{it'}$ wherever
the derivative exists and set $a_i(t')=0$ on the remaining null set.
This defines a measurable speed path in $[0,1]$. The projection
inclusion gives $\dot A_{jt'}\in K[\sigma_j](S_{t'})$ almost everywhere,
while feasibility gives $\dot X_{t'}=g$. These are exactly the conditions
in \cref{eq:generalized-deviation}.
\end{proof}

For verifying that a profile is an MPE, we develop a test using paths instead of r alternative strategies.  Fix any measurable speed path
$a_i:[0,\infty)\to[0,1]$ and consider every absolutely continuous state path starting at $S_0=S$ that satisfies, almost everywhere,
\begin{equation*}\tag{D}\label{eq:generalized-deviation}
\dot A_{it}=a_i(t),\text{  }
\dot A_{jt}\in K[\sigma_j](S_t), \text{ and } 
\dot X_t=g. \end{equation*}
Condition \cref{eq:generalized-deviation} restricts the resulting state path, given the chosen speed path and the opponent's strategy;
it imposes no optimality requirement on $a_i$.  By \cref{spe:lem:paths}, every history-dependent deviation against a Markov
opponent satisfies \cref{eq:generalized-deviation} with its realized own speed as the control.

\begin{lemma}[Verification with several paths]
\label{lem:generalized-verification}
Suppose finite candidate values $\widehat V_i(S)$ satisfy the following conditions at every
state $S$:
\begin{enumerate}[label=(\roman*),leftmargin=2em]
\item \emph{Attaining the candidate payoffs.} One common admissible path under $\sigma$ gives each firm $\widehat V_i(S)$.
\item \emph{No profitable deviation.} Every path from $S$ satisfying
\cref{eq:generalized-deviation}, for every measurable control $a_i$,
gives firm $i$ at most $\widehat V_i(S)$.
\end{enumerate}
Then $V_i(\sigma,S)=\widehat V_i(S)$, and $\sigma$ is an \MPE{}.
\end{lemma}

\begin{proof}
We first identify the profile's payoff, then bound strategy deviations.
Condition (i) gives $V_i(\sigma,S)\geq\widehat V_i(S)$.
Every joint admissible path satisfies \cref{eq:generalized-deviation}
with $a_i=\dot A_i$, because projection of the joint Krasovskii set
onto the opponent's coordinate lies in $K[\sigma_j](S)$.
Condition (ii) therefore bounds every such path by $\widehat V_i(S)$,
giving equality.

Now fix a feasible history $h_t$ and a deviation $\sigma_i'$ that may
depend on the full history. By \cref{spe:lem:paths}, every continuation
under $(\sigma_i',\sigma_j)$ satisfies \cref{eq:generalized-deviation}
from $S_t$. Condition (ii) bounds its payoff by $\widehat V_i(S_t)$.
Taking suprema and using the fact that Markov continuation payoffs
depend only on the current state gives $
V_i((\sigma_i',\sigma_j),h_t)\leq\widehat V_i(S_t)
=V_i(\sigma,h_t),$ which establishes \cref{def:mpe}.
\end{proof}

\begin{lemma}[Racing guarantee]\label{lem:race-baseline}
If a firm starts as a weak leader at distance $D$ above safety and then races,
it earns at least $R(D)$ per unit of its initial flow against every
admissible response.
\end{lemma}

\begin{proof}
The firm remains a weak leader because its rival cannot develop
faster. Its frontier is its initial capability plus elapsed time,
so its hazard is exactly $\lambda(D+(1-g)t)$. Its flow grows at
least at rate $\alpha-\beta$, since its rival's speed is at most one.
Integrating gives $R(D)$. Prescribing speed one at every state or history fixes the realized
speed at one almost everywhere under both (A) and (HA).
\end{proof}
\subsection{Proof of Lemma 1}
\begin{proof}[Proof of \cref{lem:markov-best-response}]
Fix a Markov opponent $\sigma_j$.  

\emph{Step 1: Bound the speeds available to deviations.}
From \cref{spe:lem:paths}, the opponent's speed along any admissible
deviation path lies in $K[\sigma_j](S_t)$, the set of speeds allowed
by its Markov strategy under regularization. Thus
$
\dot A_{it}\in[0,1]$ and $
\dot A_{jt}\in K[\sigma_j](S_t)$ almost everywhere. 

\emph{Step 2: Construct a strategy that allows the rectangle.}
We seek a measurable Markov strategy $\sigma_i^*$ and a countable
exceptional set of states $E$ such that, ordering speeds as firm $i$
then firm $j$,
\begin{equation*}\tag{RECT}\label{eq:markov-response-rectangle}
K[(\sigma_i^*,\sigma_j)](S)
=[0,1]\times K[\sigma_j](S),\qquad S\notin E.
\end{equation*}

Consider the opponent's graph
$\{(S,\sigma_j(S)):S\in\mathcal S\}$ in state--speed space.
List all open balls with rational centers and rational radii
as $U_1,U_2,\ldots$. These sets form a \emph{countable base}:
every neighborhood of any point contains a listed ball that
also contains that point.

At stage $n$, inspect the graph points inside the ball $U_n$.
If there are only finitely many, skip that ball. Otherwise, there
are infinitely many distinct states to choose from, because the
opponent prescribes just one speed at each state. Earlier stages
have used only finitely many states:
\[
\#\{\text{states chosen before stage }n\}\leq2(n-1)<\infty.
\]
We can therefore choose two graph points in $U_n$ whose states
have not been chosen earlier. Put the first state in $E_0$, the
set receiving prescription zero, and the second in $E_1$, the
set receiving prescription one. After processing all balls, the
sets $E_0$ and $E_1$ are countable and disjoint. Define our Markov
strategy by
\[
\sigma_i^*(S)=\ind\{S\in E_1\}.
\]
This strategy prescribes one on $E_1$ and zero everywhere else,
including on $E_0$. It is measurable because $E_1$ is countable.

For every accumulation point $(S,b)$ of the opponent's graph
and each own prescription $a\in\{0,1\}$, there are states
$S_n\in E_a$ such that
\[
S_n\longrightarrow S
\quad\text{and}\quad
\bigl(\sigma_i^*(S_n),\sigma_j(S_n)\bigr)\longrightarrow(a,b).
\]
Indeed, every neighborhood of an accumulation point contains
infinitely many graph points. A listed ball containing that
point and contained in the neighborhood therefore receives a
state of each type.

Take $E$ to be the state coordinates of isolated graph points.
This set is countable because each isolated point has a listed
ball containing no other graph point. Outside $E$, every point
over $S$ in the graph's closure is an accumulation point.

It remains to show that the construction gives
\cref{eq:markov-response-rectangle}. Fix a state $S\notin E$.
The opponent's allowed speeds form an interval with lowest
speed $b_-$ and highest speed $b_+$:
\[
K[\sigma_j](S)=[b_-,b_+].
\]
The endpoints are the limits of the infima and suprema of
prescribed speeds on shrinking state neighborhoods, so each
is approached by opponent prescriptions at states converging
to $S$. The construction gives all four corners:
\[
\{0,1\}\times\{b_-,b_+\}
\subseteq K[(\sigma_i^*,\sigma_j)](S).
\]
For either endpoint $b$, convex combinations allow every own
speed $a\in[0,1]$. Combining the two opponent endpoints then
fills the rectangle, including when $b_-=b_+$. Conversely,
every allowed joint speed has its first coordinate in $[0,1]$
and its opponent coordinate in $K[\sigma_j](S)$.
This proves \cref{eq:markov-response-rectangle}.

Because the graph of $\sigma_j$ is a separable metric space, its
nonisolated part contains two disjoint countable subsets such that
each meets every basic open set containing infinitely many graph
points. Let $E_0$ and $E_1$ be their state projections, and prescribe
zero on $E_0$, one on $E_1$, and zero elsewhere. The isolated graph
points form a countable set $E$. At every $S\notin E$, each limiting
opponent speed can be approached along states receiving either own
prescription, so convexification yields
\[
K[(\sigma_i^*,\sigma_j)](S)
=[0,1]\times K[\sigma_j](S).
\]
\emph{Step 3: Compare paths and payoffs.}
Fix a feasible history $h_t$ and a deviation $\sigma_i'$ that may depend
on the full history. Take any capability path allowed by that deviation.
From \cref{spe:lem:paths}, this path obeys the speed bounds in Step 1
and safety continues to advance at speed $g$, as stated in
\cref{eq:generalized-deviation}. Hence
$
X_s=X_t+g(s-t),s\geq t.$ Since $g>0$, the path visits each full state at most once and reaches
the countable set $E$ at only countably many dates. At almost every other
date, \cref{eq:markov-response-rectangle} allows the same speeds under
$(\sigma_i^*,\sigma_j)$. From
\cref{spe:lem:paths}, the history-based and state-based path rules
coincide for Markov profiles. The path we have just checked using
the state-based rule is therefore also admissible after the original
history $h_t$. Thus our strategy allows the same path:
$
\mathcal P((\sigma_i',\sigma_j);h_t)
\subseteq\mathcal P((\sigma_i^*,\sigma_j);h_t).
$ Profits and disaster risk are the same on the same path. Taking the
supremum over allowed paths therefore gives
$
V_i((\sigma_i',\sigma_j),h_t)
\leq V_i((\sigma_i^*,\sigma_j),h_t)
=V_i((\sigma_i^*,\sigma_j),S_t).
$
Since the deviation was arbitrary, our Markov strategy is a best
response after the history $h_t$. The same strategy works after every
feasible history. 

\Needspace{6\baselineskip}
Finally, if a Markov profile $\sigma$ passes every Markov deviation
test, then
\[
V_i((\sigma_i',\sigma_j),h_t)
\leq V_i((\sigma_i^*,\sigma_j),S_t)
\leq V_i(\sigma,S_t)=V_i(\sigma,h_t),\] 
where the second inequality is the test applied to $\sigma_i^*$.
Thus the profile is an \MPE. The converse holds because every Markov
deviation is also an allowed history-dependent deviation.
\end{proof}

\subsection{Proof of Theorem 1}
\label{sec:equilibrium-proof}\label{sec:general-risk-proof}
\leavevmode\par

\begin{figure}[H]
\centering
\input{figures/cooperating_v3_proof_map.tex}
\caption{The equilibrium construction in \cref{thm:main}.  Numbered steps
belong to \cref{lem:pacing-verification}. The additional uniqueness
claim in part (a) uses the aggregate-profit argument in
\cref{sec:spe-proofs}.}
\label{fig:proof-map}
\end{figure}
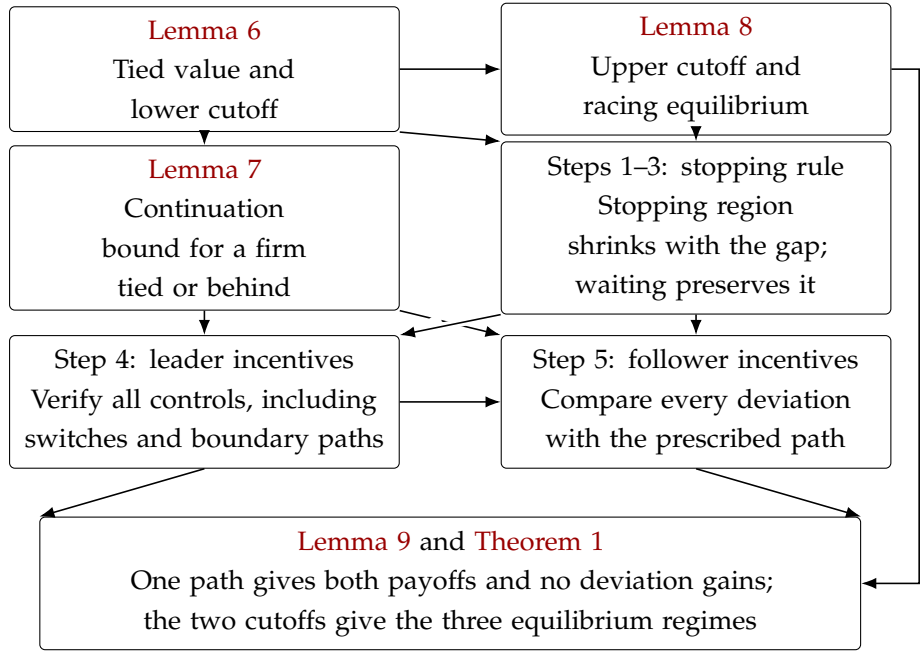

We start by introducing an auxillary problem where a single decision-maker chooses a speed path  $(a_t)_t$. 

Fix an initial distance from safety $D\in\mathbb R$,
a flow growth coefficient $\eta>0$, and a discount rate
$\rho>\eta$. The auxillary problem is: 
\[
\begin{aligned}
\max_{\substack{a:[0,\infty)\to[0,1]\\ a\text{ measurable}}}
\quad&
\int_0^\infty
\exp\left\{
-\rho t+\eta\int_0^t a(s)\,\de s
-\int_0^t\lambda(D_s)\,\de s
\right\}\,\de t\\
\text{subject to}\quad&
\dot D_t=a(t)-g \quad\text{for almost every }t\geq0,\\
& D_0=D.
\end{aligned}
\]

The value of pacing on safety is $P_\eta(0)=1/(\rho-g\eta)$.
For this auxiliary problem, denote the racing and waiting-then-pacing
values by $R_\eta$ and $P_\eta$ (fixing $\rho)$: 
\begin{align*}
R_\eta(D)&=\int_0^\infty\exp\left\{-(\rho-\eta)t
 -\int_0^t\lambda(D+vs)\de s\right\}\de t,\\
P_\eta(D)&=\int_0^{D/g}\exp\left\{-\rho t
 -\int_0^t\lambda(D-gs)\de s\right\}\de t
 \\&\quad+\exp\left\{-\rho D/g-\frac1g\int_0^D\lambda(s)\de s\right\}P_\eta(0).
\end{align*}
The racing value is defined at every $D$; the waiting formula
applies for $D\geq0$.

\begin{lemma}[When to race and when to wait]\label{lem:scalar-control}\leavevmode
\begin{enumerate}[label=(\alph*),leftmargin=1.8em]
\item  
If $R_\eta(0)\geq P_\eta(0)$, then perpetual racing is optimal
for the auxiliary problem from every initial $D$. If the
inequality is strict, every optimal path races almost everywhere.
At equality, pacing is also optimal when starting on safety.

\item  If
$
R_\eta(0)<P_\eta(0),
$
then the racing and waiting values have a unique positive crossing.
The optimal value in the auxiliary problem above safety is $\max\{P_\eta(D),R_\eta(D)\}$,
attained by the following choices:

\begingroup
\small
\renewcommand{\arraystretch}{1.15}
\begin{tabular}{@{}p{.36\linewidth}@{\hspace{1em}}p{\dimexpr.64\linewidth-1em\relax}@{}}
\textbf{Initial position} & \textbf{Optimal continuation}\\
\hline
Below safety & Race to safety, then pace.\\
On safety & Pace.\\
Above safety but below the crossing & Pause until safety arrives, then pace.\\
At the crossing & Either pause until safety arrives and then pace, or race forever.\\
Above the crossing & Race forever.
\end{tabular}
\endgroup
\end{enumerate}
\end{lemma}

\begin{proof}
We check racing and waiting separately, locate their crossing, and
then bound every feasible path, including time spent at safety or the
crossing.

\emph{Step 1: the payoff inequality.}
At a differentiable state, a candidate value $F(D)$ must satisfy
$
1+(a-g)F'(D)+(\eta a-\rho-\lambda(D))F(D)\leq0
\quad\text{for every }a\in[0,1].$ 
We call the left side the \emph{payoff residual}: it is current flow
plus the change in continuation value, net of discounting and disaster
risk. A nonpositive residual gives an upper bound on payoff after
integration. Its coefficient on speed is $F'+\eta F$, so it suffices
to check speeds zero and one. For $D>0$, the equations for racing and
waiting are
$
vR_\eta'=(\rho-\eta+\lambda)R_\eta-1
\quad\text{and}\quad
gP_\eta'=1-(\rho+\lambda)P_\eta.
$
Both values are nonincreasing, and $P_\eta$ is strictly decreasing
on $(0,\infty)$. For waiting, write the development increment as
$(gt-D)_+$: increasing $D$ reduces development and increases exposure
at each date. For racing, monotonicity follows from the integral.
For each $t\geq0$, let $f_t(D):=\exp\left\{-(\rho-\eta)t
-\int_0^t\lambda(D+vs)\,\de s\right\}$ be the integrand of $R_{\eta}(D)$. 
Concavity of $\lambda$ on $(0,\infty)$ makes
$\log f_t$ convex there. Thus, for $D_1,D_2>0$ and
$\theta\in(0,1)$, H\"older's inequality gives
\[
\begin{aligned}
R_\eta(\theta D_1+(1-\theta)D_2)
&\leq \int_0^\infty
f_t(D_1)^\theta f_t(D_2)^{1-\theta}\,\de t\\
&\leq
\left(\int_0^\infty f_t(D_1)\,\de t\right)^\theta
\left(\int_0^\infty f_t(D_2)\,\de t\right)^{1-\theta} =R_\eta(D_1)^\theta R_\eta(D_2)^{1-\theta}.
\end{aligned}
\]
Hence $\log R_\eta$ is convex, so its derivative
$R_\eta'/R_\eta$ is nondecreasing on $(0,\infty)$.

\emph{Step 2: when racing is optimal everywhere.}
Suppose $R_\eta(0)\geq P_\eta(0)$. The racing speed
coefficient $R_\eta'+\eta R_\eta$ has right limit at zero
$
\dfrac{(\rho-g\eta+\lambda_0)R_\eta(0)-1}{v}\geq0.
$ Since $R_\eta>0$, the coefficient has the sign of
$R_\eta'/R_\eta+\eta$. Log convexity keeps that sign nonnegative
above safety.
Below safety, racing reaches $D=0$ after $-D/v$ units of time.
Writing $z=\exp\{(\rho-\eta)D/v\}$ for discounted profit growth
until that date, the speed coefficient is
\[
\frac{\eta(1-z)}{\rho-\eta}
+\frac zv\bigl[(\rho-g\eta)R_\eta(0)-1\bigr]\geq0.
\]
On safety the pacing residual is
$1-(\rho-g\eta)R_\eta(0)\leq0$.
These inequalities verify global racing. If $R_\eta(0)=P_\eta(0)$,
pacing on safety is also optimal.

\emph{Step 3: when waiting and racing have a positive crossing.}
Suppose $R_\eta(0)<P_\eta(0)$. Above safety, we show that the value
is $\max\{P_\eta(D),R_\eta(D)\}$, with a unique positive crossing.
Let the speed coefficient for the waiting value be
$Q_\eta=P_\eta'+\eta P_\eta$. Differentiation gives

$gQ_\eta'=-(\rho+\lambda)Q_\eta+\eta-\lambda'P_\eta$ and $
Q_\eta(0+)=-\lambda_0P_\eta(0)/g\leq0.$ 
The term $\eta-\lambda'P_\eta$ is nondecreasing: concavity
makes $\lambda'$ nonnegative and nonincreasing, and $P_\eta$ decreases.
This term tends to $\eta$. If $\lambda$ is bounded, its derivative
tends to zero; if it is unbounded, $P_\eta$ tends to zero by dominated
convergence and $\lambda'$ is bounded on $D\geq1$. An integrating factor implies that
$Q_\eta$ can cross zero only from negative to positive. To see that it  eventually becomes positive, take $D$ large enough that $\eta-\lambda'(D)P_\eta(D)\geq\eta/2$. Whenever $Q_\eta(D)\leq0$,
the differential equation then gives
$
gQ_\eta'(D)
=-(\rho+\lambda(D))Q_\eta(D)+\eta-\lambda'(D)P_\eta(D)
\geq\eta/2.
$ If $Q_\eta$ is still negative, it must reach zero at some finite $D$ becasue derivative is bounded below by $\eta/(2g)$ until it does.
At every subsequent zero its derivative is strictly positive, so it cannot cross back to negative values. 

Define the integrating factor
$M_\eta(D)=\exp\{[(\rho-\eta)D+\int_0^D\lambda(s)\de s]/v\}$.
Subtracting the two control equations yields
$
\left(\dfrac{P_\eta-R_\eta}{M_\eta}\right)'
=\dfrac{Q_\eta}{vM_\eta}.$ Both values are bounded by $1/(\rho-\eta)$ and $M_\eta$ diverges,
so $(P_\eta-R_\eta)/M_\eta$ tends to zero. Its derivative is eventually
positive, so the ratio is eventually negative. Since it starts positive
and its derivative can change sign only from negative to positive,
it crosses zero exactly once, with $Q_\eta<0$ at the crossing.
The waiting inequality holds below this crossing. At the crossing, $
Q_\eta=-\dfrac vg(R_\eta'+\eta R_\eta)<0.$ Because $R_\eta'/R_\eta$ is nondecreasing, the racing speed
coefficient stays positive above the crossing.
Below safety, racing to safety and then pacing has development
coefficient $\eta(1-z)/(\rho-\eta)\geq0$. On safety, speed $g$
attains $P_\eta(0)$ with zero hazard.

\emph{Step 4: integrate the bound along every path.}
The candidate is locally Lipschitz and continuously differentiable
except at safety and, when present, the positive crossing. Its
composition with any absolutely continuous state path is absolutely
continuous. Away from those states, the chain rule gives the verified
payoff inequality. At almost every date when the path is at either of these states,
its state derivative is zero, so its speed is $a=g$. On safety, the safety
residual applies with hazard zero. At the positive crossing, the
residual is
$
1-(\rho-g\eta+\lambda)P_\eta=gQ_\eta\leq0,$ 
so the inequality holds almost everywhere along every path.
The discounted terminal value vanishes because $\rho>\eta$ and the
candidate is bounded by $1/(\rho-\eta)$. The stated policies attain
the bound. When $R_\eta(0)>P_\eta(0)$, the racing speed coefficient
is strictly positive off safety and the safety residual is strictly
negative. Equality then requires racing almost everywhere.
\end{proof}

\begin{lemma}\label{lem:tied-bound}
If a Markov opponent races whenever it is strictly behind, then a
firm starting weakly behind remains weakly behind on every path
satisfying \cref{eq:generalized-deviation}. Its continuation payoff
is at most $e^{cA_i}\max\{P(A_i-X),R(A_i-X)\}$.
\end{lemma}

\begin{proof}
Whenever the deviator is strictly ahead, its opponent's speed is one
even after scalar regularization. The positive part of
their capability difference therefore has nonpositive derivative
almost everywhere and starts at zero, so it remains zero.
Along every continuation, the deviator's flow is at most $e^{cA_i}$
and its hazard is at least $\lambda(A_i-X)$. Applying
\cref{lem:scalar-control} with $(\eta,\rho)=(c,r)$ to its own
capability path gives the bound.
\end{proof}

For the leader facing a racing follower, write
$B=r+\beta-g\alpha$ for the discount rate net of profit growth
while the leader paces. It is positive, and $B>r-gc=1/P(0)$.

\begin{lemma}[The racing cutoffs]\label{lem:racing-cutoffs}
Global racing is an \MPE{} if and only if $R(0)\geq1/B$.
If $R(0)>P(0)$, no \MPE{} outcome paces forever from a tie on safety.
\end{lemma}

\begin{proof}
We rule out pacing when its payoff is too low, then check when racing
is a best response to a racing rival. If $R(0)>P(0)$, the racing
guarantee in \cref{lem:race-baseline} gives a profitable deviation
from perpetual pacing at a tie on safety.

Apply \cref{lem:scalar-control} with $(\eta,\rho)=(\alpha,r+\beta)$. The racing
value is the same $R$ and the safe value is $1/B$. This is the leader's
problem when its follower never catches up. If $R(0)\geq1/B$, global racing is an \MPE{}. Against a rival that
always races, an initially leading firm's payoff per unit of its initial flow
is bounded by this auxiliary problem: its flow grows at rate
$\alpha a-\beta$, and actual hazard is at least $\lambda(A_i-X)$.
Racing attains the bound $R(D)$. An initially trailing firm cannot
catch a racing rival; full speed maximizes its flow without changing
the frontier or hazard. At a tie the same bound or pointwise comparison
applies. These arguments bound every measurable deviation, and the
common racing path attains both payoffs.

Conversely, if $R(0)<1/B$, start with a strict leader on safety.
Against global racing it gains by pacing for $0<h<G/v$ and then
resuming racing. Per unit of the leader's initial flow, the gain is
$
(1-e^{-Bh})\bigl[1/B-R(0)\bigr]>0.
$ The pacing interval incurs zero hazard, including when $\lambda_0>0$.
Thus global racing is an \MPE{} exactly when
$\lambda^{\mathrm{eff}}\leq v\alpha$.\end{proof}

\begin{lemma}[Pacing passes the verification test]\label{lem:pacing-verification}
If $R(0)\leq P(0)$, then $\sigma^P$ is an \MPE{} with finite payoffs
and a common attaining path from every state.
\end{lemma}
\begin{proof}
If $R(0)=P(0)$, prescribe speed one at strict gaps and ties off
safety, and speed $g$ at safety ties. A strict leader's deviations
are bounded by $R(D)$ up to the first tie by
\cref{lem:scalar-control}, applied with
$(\eta,\rho)=(\alpha,r+\beta)$, and thereafter by
\cref{lem:tied-bound}. A follower cannot catch a racing leader,
so racing maximizes its flow at the given hazard. At ties,
\cref{lem:tied-bound} again gives the bound $R(D)$.
The prescribed paths attain these bounds, so
\cref{lem:generalized-verification} applies. This is the construction
with $\bar D=0$.

Henceforth suppose $R(0)<P(0)$, so $\bar D>0$. Write $P_\alpha$
for the waiting-then-pacing value in \cref{lem:scalar-control}
with $(\eta,\rho)=(\alpha,r+\beta)$.

\emph{Step 1: accommodation values and the stopping formula.}
The values $e^{-\beta G}V_L(D,G)$ and $C(D,G)$ are both measured
per unit of the leader's current flow. For $0\leq D\leq\bar D$,
put $T=D/g$ and
$E_D(t)=\exp\{-(r+\beta)t-\int_0^t\lambda(D-gs)\de s\}$.
Immediate accommodation gives
\begin{align*}
C(0,G)&=\frac1B+\left(P(0)-\frac1B\right)e^{-BG/v},\\
C(D,G)&=\int_0^G E_D(t)\de t+E_D(G)P(D-gG),
&&G\leq T,\\
C(D,G)&=\int_0^T E_D(t)\de t+E_D(T)C(0,G-T),
&&G\geq T.
\end{align*}
The two formulas for $C(D,G)$ and their first derivatives agree
at $G=D/g$. Moreover, $C$ decreases in both arguments:
its integrand is
\[
e^{-rt+c(gt-D)_+-\beta\min\{G,t-(gt-D)_+\}
-\int_0^t\lambda(D-gs)\de s},
\]
which decreases with either $D$ or $G$. Also
$C(D,G)\geq P_\alpha(D)$.

Define
\[
M(D)=\exp\left\{\frac{(r-c)D+\int_0^D\lambda(s)\de s}{v}\right\},
\qquad
H(D,G)=\frac{C(D,G)-R(D)}{M(D)}.
\]
The stopping formula in \cref{sec:tied-values} becomes
\[
e^{-\beta G}V_L(D,G)
=R(D)+M(D)\max\left\{0,\max_{D\leq d\leq\bar D}H(d,G)\right\}.
\]
Choose the earliest finite maximizer whenever one is optimal.

Along any path in the stopping construction, the leader's payoff
per unit of initial flow has integrand
$e^{-rt+c(A_t-A_0)-\Lambda_t}
e^{-\beta\min\{G,t-(A_t-A_0)\}}$.
Dropping the second factor and applying
\cref{lem:scalar-control}, with \cref{lem:tied-bound} after
catch-up, gives the upper bound below. Racing and immediate
accommodation give the lower bound:
\[
\max\{P_\alpha(D),R(D)\}
\leq e^{-\beta G}V_L(D,G)
\leq\max\{P(D),R(D)\}.
\]
For $D\geq\bar D$, racing attains the upper bound $R(D)$.

\emph{Step 2: monotonicity in the gap and the safety residual.}
The waiting equation and its speed coefficient are
\begin{align*}
0&=1-gC_D-C_G-(r+\beta+\lambda(D))C,\\
Q_C&=C_D+C_G+\alpha C
=\frac{1-vC_G-(B+\lambda(D))C}{g}.
\end{align*}
For $0<D<\bar D$, differentiation gives
\[
(Q_C)_G=
\begin{cases}
\displaystyle\frac{\beta E_D(G)}g
\bigl[(r-gc+\lambda(D))P(D-gG)-v\bigr],&G<D/g,\\[1ex]
\displaystyle-\frac{\lambda(D)}gC_G,&G>D/g.
\end{cases}
\]
Both expressions are positive. For the first, the scalar waiting
inequality gives $(r-gc+\lambda(d))P(d)\geq1$ at $d=D-gG$,
and $\lambda(D)\geq\lambda(d)$. For the second, $C_G<0$ and
$\lambda(D)>0$; concavity and monotonicity imply the latter because
$R(0)<P(0)$ excludes identically zero risk. Since $Q_C$ is
continuous at $G=D/g$,
$H_D=Q_C/(vM)$, $H_{DG}>0$ wherever defined, and $H_G<0$.

When $\phi=0$, differentiation also gives $(Q_C)_D>0$ for
$G>0$, so $H(\cdot,G)$ has no interior local maximum.
Since $C(\bar D,G)<P(\bar D)=R(\bar D)$, a finite optimal stop
must then be immediate, recovering the choice between immediate
accommodation and perpetual racing.

The stopping region is defined by $H(D,G)\geq0$ and
$H(D,G)\geq H(d,G)$ for every $d\in[D,\bar D]$.
Since $H(D,G)$ decreases with $G$ and
$H(d,G)-H(D,G)$ increases with $G$ for $d>D$,
increasing the gap cannot turn a racing state into a stopping
state. Thus the rule that pauses in this closed region and races
outside is nondecreasing in $G$.

On safety, define
$J(G)=e^{-\beta G}V_L(0,G)-C(0,G)$. The stopping formula gives
\[
J(G)=\max\left\{0,\ R(0)-C(0,G),\
\max_{0\leq d\leq\bar D}[H(d,G)-H(0,G)]\right\}.
\]
Since $C_G<0$ and $H$ has increasing differences, $J$ is
nonnegative, nondecreasing, and locally Lipschitz. Using
$1-vC_G(0,G)-BC(0,G)=0$, the safety pacing residual is
\[
e^{\beta G}-v\partial_GV_L(0,G)-(r-gc)V_L(0,G)
=-e^{\beta G}[vJ'(G)+BJ(G)]\leq0
\]
almost everywhere. If $J(G)=0$, pacing reduces the gap and keeps
$J$ zero until the tie. Otherwise the earliest optimal choice is
a positive target or perpetual racing. The safety rule is
therefore also nondecreasing in $G$.

If $R(0)\leq1/B$, safety is a stopping state at every finite gap.
Immediate pacing beats perpetual racing because
$C(0,G)>1/B\geq R(0)$.
To compare pacing with racing to a positive target $d$ and then
accommodating, note that $H(d,G)-H(0,G)$ increases with $G$.
As $G\to\infty$, catch-up disappears and $C(d,G)\to P_\alpha(d)$,
so this difference converges to
$(P_\alpha(d)-R(d))/M(d)-[1/B-R(0)]$.
This is the gain from racing to $d$ and then waiting and pacing,
rather than pacing immediately, in the auxiliary problem with
$(\eta,\rho)=(\alpha,r+\beta)$.
By \cref{lem:scalar-control}, that gain is nonpositive.
Hence $H(d,G)\leq H(0,G)$ at every finite gap.

\emph{Step 3: accommodation remains optimal along the prescribed path.}
Start at a stopping state $(D_0,G_0)$ with $D_0,G_0>0$.
While waiting, $\dot D=-g$ and $\dot G=-1$ until safety or a tie,
and
\[
\frac{\de}{\de t}Q_C(D_t,G_t)
=(r+\beta+\lambda(D_t))Q_C(D_t,G_t)
+\lambda'(D_t)C(D_t,G_t)-\alpha.
\]
At $G=D/g$, the derivative along the waiting path is given by
the same expression from either formula for $C$.
Both $C(D_t,G_t)$ and $\lambda'(D_t)$ increase during waiting.
Thus the derivative of $Q_C$ multiplied by
$\exp\{-\int_0^t(r+\beta+\lambda(D_s))\de s\}$
can change sign only from negative to positive, so the transformed
coefficient is maximized at an endpoint. Initially $Q_C\leq0$
by optimality against larger targets. Its terminal limit is
$-\lambda_0C(0,G)/g\leq0$ at safety, or
$P'(D)+cP(D)\leq0$ at a tie below $\bar D$.
Hence $Q_C\leq0$ throughout waiting.

At a later waiting state $(D_1,G_1)$, the gap when the path passed
$d\in[D_1,D_0]$ was $G_1+(d-D_1)/g\geq G_1$.
Since $(Q_C)_G>0$, $Q_C(d,G_1)\leq0$, so
$H(D_1,G_1)\geq H(d,G_1)$ throughout this interval.
For $d\in[D_0,\bar D]$, increasing differences and initial
optimality give $H(d,G_1)\leq H(D_0,G_1)$, while
$H(D_0,G_1)\geq H(D_0,G_0)\geq0$.
Thus $(D_1,G_1)$ remains a stopping state.
Continuity extends this conclusion to safety, where monotonicity
of $J$ keeps pacing optimal until catch-up.

\emph{Step 4: verify the leader's best response.}
The racing and waiting residuals are, respectively,
$e^{\beta G}+v\partial_DV_L-(r-c+\lambda)V_L$ and
$e^{\beta G}-g\partial_DV_L-\partial_GV_L-(r+\lambda)V_L$.
Their difference is $e^{\beta G}Q$, where
$Q=e^{-\beta G}(\partial_DV_L+\partial_GV_L+cV_L)$.
In the stopping region, the waiting residual is zero and the
racing residual is $e^{\beta G}Q_C\leq0$.
Where $V_L=e^{\beta G}R$, the racing residual is zero;
$P_\alpha\leq R$ and \cref{lem:scalar-control} imply
$R'+\alpha R\geq0$, so the waiting residual is nonpositive.

Fix a gap $G$ at which the maximum over targets in the stopping
formula is differentiable with respect to $G$, and consider a
racing interval where $V_L>e^{\beta G}\max\{C,R\}$.
Such differentiability holds for almost every gap.
On this interval,
$V_L(D,G)=e^{\beta G}[R(D)+M(D)Z(G)]$,
where $Z$ is the maximum of $H(d,G)$ over a fixed target interval
bounded away from the current $D$. The maximum is locally
Lipschitz; wherever differentiable, every maximizing target has
$G$-derivative $Z'(G)$. The racing residual is zero.

Let $d_*$ be the optimal target ending the interval. It is
interior because $C(\bar D,G)<R(\bar D)$ for $G>0$.
Optimality and the envelope identity give $H_D(d_*,G)=0$,
$\partial_GV_L=e^{\beta G}(C_G+\beta C)$,
$Q(d_*,G)=0$, and
$\partial_DV_L(d_*,G)=e^{\beta G}C_D(d_*,G)\leq0$.
Along the racing interval,
$
v\partial_{DD}V_L
=(r-c+\lambda)\partial_DV_L+\lambda'V_L,$ and $
v\partial_DQ
=(r-c+\lambda)Q+\lambda'e^{-\beta G}V_L-\alpha.$ 
Integrating the first equation backward gives
$\partial_DV_L\leq0$. Consequently
$\lambda'e^{-\beta G}V_L-\alpha$ is nonincreasing in $D$, and
$
\partial_D(Q/M)
=\dfrac{\lambda'e^{-\beta G}V_L-\alpha}{vM}$ 
changes sign at most once, from positive to negative.
Thus $Q/M$ is minimized at an endpoint.

Let $d_-$ be the left endpoint. If $d_->0$, Step 3 keeps
accommodation optimal at $(d_--gt,G-t)$ for small $t\geq0$.
Racing to the fixed target $d_*$ remains feasible there, so
$
R(d_--gt)+M(d_--gt)H(d_*,G-t)
\leq C(d_--gt,G-t),
$ with equality at $t=0$. Right differentiation shows that this
fixed-target payoff has nonpositive waiting residual.
Its racing residual is zero, and its derivatives agree with
those of the envelope at the chosen gap, so $Q(d_-,G)\geq0$.
If $d_-=0$, the safety residual from Step 2 gives
\[
vQ(0+,G)=e^{-\beta G}
\left\{\lambda_0V_L(0,G)
-\left[e^{\beta G}-v\partial_GV_L(0,G)
-(r-gc)V_L(0,G)\right]\right\}\geq0.
\]
Hence $Q\geq0$ throughout the interval, and the waiting residual
is nonpositive. Fubini's theorem gives both residual inequalities
almost everywhere on $D>0,G>0$.

The payoff residual at any speed $a\in[0,1]$ is a convex
combination of the racing and waiting residuals.
To cover paths on switching boundaries, smooth $V_L$ with a
nonnegative mollifier of radius $\varepsilon$.
The derivative coefficients are constant and the flow and hazard
are locally Lipschitz, so both payoff residuals computed using
the smoothed value are at most $O(\varepsilon)$ uniformly on
compact subsets of $D>0,G>0$.
Integrating along any absolutely continuous control path and
letting $\varepsilon\downarrow0$ proves the verification inequality.
Apply this up to finite dates before a tie, then use continuity.

Below safety,
$
V_L(D,G)
=e^{\beta G}\frac{1-e^{(r-c)D/v}}{r-c}
+e^{(r-c)D/v}V_L(0,G).$ Its racing residual is zero, and
\[
Q(D,G)=\frac{\alpha(1-e^{(r-c)D/v})}{r-c}
+\frac{e^{(r-c)D/v}}v
\left\{e^{-\beta G}
[(r-gc)V_L(0,G)+v\partial_GV_L(0,G)]-1\right\}\geq0
\]
almost everywhere by Step 2. The same smoothing argument verifies
the inequality along paths in $D<0$.
On $D=0$, the leader's speed is $g$ almost everywhere, and the
safety residual applies; smoothing $V_L(0,\cdot)$ covers its
nondifferentiability set. Concatenate across these regions as in
\cref{lem:scalar-control}. At a tie, \cref{lem:tied-bound}
bounds every continuation by $\max\{P(D),R(D)\}$.

Integrate to the first tie or infinity. The bound
$V_L(D,G)\leq e^{\beta G}/(r-c)$ makes the discounted terminal
term vanish: before a tie the gap cannot increase, the frontier
grows at speed at most one, and $r>c$.
By Step 3 and the prescribed optimal common-speed path after
a tie, the earliest optimal stopping path attains the bound.

\emph{Step 5: verify the follower's best response.}
Use $b(D,G)$ for strict leaders, $a^*(D)$ at ties, and speed one
for strict followers.

For an initial strict follower, write $\bar A_t$ and $\bar A_t^F$
for the prescribed frontier and follower capabilities. The leading
firm races to its earliest optimal stopping point, then waits for
safety and paces; if a tie occurs while waiting, both firms complete
that continuation. If no finite target is optimal, the leading firm
races forever. The prescribed follower satisfies
$\bar A_t^F=\min\{A_0^F+t,\bar A_t\}$.

Fix any unilateral deviation by the follower, with the opponent
following its prescribed strategy, and any  continuation
path admitted by \cref{eq:generalized-deviation}. We first show that
the opponent's capability is at least $\bar A_t$ while the deviator
remains weakly below $\bar A_t^F$ and is a strict follower. During prescribed racing, the deviator's capability is at most
$A_0^F+t$. Until the opponent departs from the prescribed path, the  gap is
at least the prescribed gap. By Step 2 the gap  stays in the open racing
region until the earliest stopping point, including at safety
crossings with $J>0$. The opponent's regularized speed is therefore one,
precluding a departure to a lower path. During the prescribed pause, the opponent's  capability
cannot decrease. During prescribed pacing, the prescribed
frontier lies on the safety boundary, which the opponent cannot
cross from above: strictly below safety its speed is one and
$\dot D=v>0$. This proves the comparison.

If the prescribed follower never catches up, then
$\bar A_t^F=A_0^F+t$ forever. The deviator has no higher capability
and faces no lower opponent capability or hazard, so its payoff
is at most the prescribed payoff.

Otherwise, let $T_C$ be the prescribed catch-up date and let
$\tau\geq T_C$ be the first date when the deviator reaches
$\bar A_t$, with $\tau=\infty$ if it never does. Before $T_C$,
its capability is at most $A_0^F+t=\bar A_t^F$; from $T_C$
until $\tau$, it is below $\bar A_t=\bar A_t^F$. It remains
a strict follower before $\tau$, since an earlier tie would
require reaching an opponent capability at least $\bar A_t$.
The comparison therefore applies throughout this interval:
its flow payoff is no greater than the prescribed follower's,
and its cumulative hazard is no smaller.

If $\tau=\infty$, this proves the payoff bound. If $\tau<\infty$,
the deviator has capability $\bar A_\tau$ and remains weakly
behind. By \cref{lem:tied-bound}, its continuation payoff is
at most $e^{c\bar A_\tau}
\max\{P(\bar A_\tau-X_\tau),R(\bar A_\tau-X_\tau)\}$.
The prescribed firms are tied at $\bar A_\tau$ and their
continuation attains this bound. Combining the continuation
bound with the preceding flow and survival comparisons proves
that no follower deviation is profitable.

\emph{Step 6: one path gives both equilibrium payoffs.}
The prescribed path, followed by the prescribed optimal
common-speed path after a tie, is jointly admissible.
Step 4 shows that it attains the leader's value, and Step 5
bounds every follower deviation by its follower payoff.
Both payoffs are finite: the leader's value is bounded in Step 4,
and the follower's flow is at most $e^{c(A_0^F+t)}$.
Thus \cref{lem:generalized-verification} establishes the
pacing equilibrium.
\end{proof}

\begin{proof}[Proof of \cref{thm:main}]
By \cref{lem:racing-cutoffs}, global racing is an \MPE{} exactly when
$\lambda^{\mathrm{eff}}\leq(1-g)\alpha$, and no equilibrium outcome
paces forever from a tie on safety when
$\lambda^{\mathrm{eff}}<(1-g)(\alpha-\beta)$.
By \cref{lem:pacing-verification}, pacing is an \MPE{} whenever
$\lambda^{\mathrm{eff}}\geq(1-g)(\alpha-\beta)$, including equality.
The two lemmas give the three regimes and paths attaining both payoffs. From a tie
on safety, the pacing path incurs zero hazard and gives each firm $P(0)$.
\end{proof}

\subsection{Proofs of the comparative statics}
\label{sec:comparative-statics-proofs}

\begin{proof}[Proof of \cref{prop:comparative-statics}]
We express effective hazard as an average over racing dates, then
compare the hazard and the weights. Write the common growth rate
as $c=\alpha-\beta$ and the speed at which a race moves away from
safety as $v=1-g$. Define the density of payoff weights by
\[
w(t):=\frac1{R(0)}
\exp\left\{-(r-c)t-\int_0^t\lambda(vs)\,\de s\right\},
\qquad t\geq0.
\]
The density integrates to one. Integrating the derivative of the
exponential gives
\begin{equation*}\tag{WH}\label{eq:weighted-hazard}
\lambda^{\mathrm{eff}}
=\int_0^\infty\lambda(vt)w(t)\,\de t.
\end{equation*}
The boundary terms are one at zero and zero at infinity, since
$r>c$. Concavity bounds $\lambda$ by an affine function at large
distances, so all integrals below are finite. The jump at zero does
not affect the integrals or the derivative at positive dates.

\emph{Step 1: Higher risk.}
A pointwise increase in $\lambda$ weakly lowers $R(0)$ and hence
weakly raises $\lambda^{\mathrm{eff}}/(1-g)$ at fixed $g$ and $r$.

\emph{Step 2: Faster safety.}
Compare racing speeds away from safety $0<v_2<v_1$, and let $w_2$
and $w_1$ be their payoff densities. Their ratio is nondecreasing:
\[
\frac{\de}{\de t}\log\frac{w_2(t)}{w_1(t)}
=\lambda(v_1t)-\lambda(v_2t)\geq0.
\]
If the densities coincide, their averages coincide. Otherwise,
because both densities integrate to one, their difference changes
sign at most once, from negative to positive. Integrating any
nondecreasing function against $w_2-w_1$ is therefore nonnegative:
subtract its value at a crossing, and each side of the crossing
has a nonnegative integrand.

Concavity and $\phi(0)=0$ give, for every $t>0$,
\[
\frac{\lambda(v_2t)}{v_2}
=\frac{\lambda_0+\phi(v_2t)}{v_2}
\geq\frac{\lambda_0+\phi(v_1t)}{v_1}
=\frac{\lambda(v_1t)}{v_1}.
\]
The function on the right is nondecreasing in $t$. Applying the
density comparison and \cref{eq:weighted-hazard} gives
\begin{align*}
\int_0^\infty\frac{\lambda(v_2t)}{v_2}w_2(t)\,\de t
&\geq\int_0^\infty\frac{\lambda(v_1t)}{v_1}w_2(t)\,\de t\geq\int_0^\infty\frac{\lambda(v_1t)}{v_1}w_1(t)\,\de t.
\end{align*}
Thus the ratio increases weakly with $g$. Under linear risk with
$\kappa>0$, the density ratio is strictly increasing and the
integrand $\lambda(v_1t)/v_1=\kappa t$ is strictly increasing,
so the last inequality is strict.

\emph{Step 3: Greater patience.}
At fixed $g$ and hazard profile, increasing $r$ multiplies the
unnormalized payoff weights by a strictly decreasing exponential
in $t$. The same crossing argument, with signs reversed, lowers
the average of the nondecreasing function $\lambda(vt)$ in
\cref{eq:weighted-hazard}. The decrease is strict when $\phi$ is
nonconstant: the density ratio is strictly decreasing and the
hazard differs on intervals of positive length. If $\phi=0$,
\cref{eq:weighted-hazard} instead gives
$\lambda^{\mathrm{eff}}=\lambda_0$. Comparing
$\lambda^{\mathrm{eff}}/v$ with $c$ and $\alpha$ in
\cref{thm:main} proves the equilibrium-region claims.
\end{proof}

\subsection{Proof of \cref{thm:spe}} \label{sec:spe-proofs}
\begin{proof}[Proof of \cref{thm:spe}(a)]
Normalize the initial safety tie to zero and put $c=\alpha-\beta$, which implies $R(0)>P(0)$ and $2\alpha\beta R(0)<c$.
Moreover, $R(D)\leq R(0)$ and, by \cref{lem:scalar-control},
$R'(D)+cR(D)>0$ for $D>0$. Define $W(G):=2P(0)(\alpha e^{\beta G}+\beta e^{-\alpha G})
/(\alpha+\beta)$ and
$\widehat V(S):=[R(\max\{D,0\})/P(0)]
e^{c\max\{A,X\}}W(\max\{A,X\}-\min_i A_i)$.
This function is nonnegative, locally Lipschitz, and equals $2R(0)$
at the initial tie.
long any feasible path, define
$\mathcal E_t:=\pi_1+\pi_2+\frac{\de}{\de t}\widehat V(S_t)
-(r+\lambda(D_t))\widehat V(S_t)$.
We show that $\mathcal E_t\leq0$ almost everywhere.

On or below safety, put $G=X-\min_i A_i$.
Raising the leader's capability to $X$ increases total flow, so
$\pi_1+\pi_2\leq e^{cX}(e^{\beta G}+e^{-\alpha G})$.
Also, $\dot G\leq g$, $W'\geq0$, and
$W/P(0)-(e^{\beta G}+e^{-\alpha G})-gW'
=[c-2\alpha\beta gP(0)](e^{\beta G}-e^{-\alpha G})
/(\alpha+\beta)\geq0$.
Since $1/P(0)=r-gc$, $\mathcal E_t$ is at most
$(1-R(0)/P(0))(\pi_1+\pi_2)<0$. This also applies almost everywhere on the safety boundary, where \(\dot A=g\).

Above safety, write $\pi_L$ and $\pi_F$ for the leader's and
follower's flows. Here
$\widehat V=2R(D)(\alpha\pi_L+\beta\pi_F)/(\alpha+\beta)$.
At a strict lead, $\mathcal E_t$ is linear in speeds, with positive
coefficient on the leader's speed and nonpositive coefficient on
the follower's. Indeed, these coefficients are respectively
$2[(R'+cR)(\alpha\pi_L+\beta\pi_F)
+\alpha\beta R(\pi_L-\pi_F)]/(\alpha+\beta)$ and
$-2\alpha\beta R(\pi_L-\pi_F)/(\alpha+\beta)$.
It is therefore maximized when the leader races and the follower
pauses. Substituting
$(1-g)R'=(r-c+\lambda)R-1$ bounds it by
$[2\alpha\beta R(D)-c](\pi_L-\pi_F)/(\alpha+\beta)<0$.
On an unsafe tie, speeds agree almost everywhere; at common speed
$a$, the same expression equals
$-2e^{cA}(1-a)(R'+cR)\leq0$, with equality only when $a=1$.

Multiply $\mathcal E_t$ by
$e^{-rt-\Lambda_t}$ and integrate. Dropping the nonnegative
terminal value and letting the horizon tend to infinity gives
total payoff at most $2R(0)$. A strict residual inequality on a set of positive measure makes total payoff strictly less than $2R(0)$. Equality therefore requires zero time on or below safety and no strict capability lead: any positive gap persists on an interval. The firms
must therefore remain tied and race almost everywhere. By \cref{lem:race-baseline}, each firm can guarantee $R(0)$ from
the initial tie, so every SPE outcome must be perpetual joint racing.
\end{proof}

\paragraph{The payoff guarantee for part (b).}
\begin{comment}
The following argument shows that if $m$ is a valid continuation
floor at every safety tie, then
\[
\widetilde L_1(\varepsilon)
+q(\varepsilon)[m-R(0)]
\]
is a valid floor before a pause of length $\varepsilon$.
The following argument shows that if $m$ is a valid continuation
floor at every safety tie, then
\[
\widetilde L_1(\varepsilon)
+q(\varepsilon)[m-R(0)]
\]
is a valid floor before a pause of length $\varepsilon$.
\end{comment}
For a pause of length
$\varepsilon>0$, define the maximum waiting time $T$ and the
guaranteed safe interval length $\Delta$ by
$
T:=\frac{1-g}{g}\varepsilon
$ and $
\Delta:=\frac{2g-1}{g(1-g)}\varepsilon.$
The first step in \cref{fig:recursive-floor} illustrates the rival's
safe option when it races throughout the initial pause. The payoff
bound also covers every other response to that pause.

After the pause, the rival can wait to safety if it is above it,
or race to safety if it is below it, and then pace. Along this
comparison its excess above safety is at most
$((1-g)\varepsilon-gs)_+$.
Write $\Gamma_\varepsilon(s)$ for the resulting exposure bound,
measuring elapsed time $s$ from the end of the pause:
\[
\Gamma_\varepsilon(s):=\frac1g
\int_{((1-g)\varepsilon-gs)_+}^{(1-g)\varepsilon}
\lambda(D)\,\de D,
\qquad s\geq0.
\]
The exposure bound is constant after $T$. If the rival races
throughout the pause, then waits, paces, and finally races from
safety at elapsed time $T+\Delta$, its continuation payoff is at least
\[
\begin{aligned}
\widetilde L_2(\varepsilon)
&:=e^{\alpha\varepsilon}
  \int_0^T e^{-(r+\beta)s-\Gamma_\varepsilon(s)}\,\de s\\
&\quad+e^{\alpha g\varepsilon-\Gamma_\varepsilon(T)}
  \int_T^{T+\Delta}e^{-(r+\beta-\alpha g)s}\,\de s\\
&\quad+e^{-\Gamma_\varepsilon(T)
  -(r-\alpha+\beta)(T+\Delta)}R(0).
\end{aligned}
\]
The payoff-ratio argument below converts this comparison into
a bound valid after every response to the pause. Define
\begin{equation*}\tag{RP}\label{spe:eq:refined-pause}
\begin{aligned}
\underline{V}_1&:=\max\left\{R(0),\ \sup_{\varepsilon>0}
     \widetilde L_1(\varepsilon)\right\}\\
\widetilde L_1(\varepsilon)
:=\int_0^\varepsilon e^{-(r+\beta)s-\Lambda_s^R}\,\de s
&+\frac{r-\alpha}{r+\beta}
 e^{-(r+\alpha+\beta)\varepsilon-\Lambda_\varepsilon^R}
 \widetilde L_2(\varepsilon).
\end{aligned}
\end{equation*}
The terminal guarantee in $\widetilde L_2$ is $R(0)$.
Using a continuation guarantee $m$ instead gives the pause bound
$\widetilde L_1(\varepsilon)+q(\varepsilon)[m-R(0)]$, where
\[
q(\varepsilon):=\frac{r-\alpha}{r+\beta}
 e^{-(r+\alpha+\beta)\varepsilon-\Lambda_\varepsilon^R
 -\Gamma_\varepsilon(T)-(r-\alpha+\beta)(T+\Delta)}.
\]
We prove this bound for any $m\in[R(0),P(0)]$ that each firm
is guaranteed after every feasible history ending at a safety tie.
The choice $m=R(0)$ remains valid when $R(0)>P(0)$.

\begin{proof}[Proof of the general pause bound]
Put $c=\alpha-\beta$ and $k=(r-\alpha)/(r+\beta)$.
Let $m\in[R(0),P(0)]$ be a normalized payoff guarantee after
every history ending at a safety tie. It also applies to a
leader on safety: pacing until the first tie at date $t$ and
then resuming equilibrium gives at least
$P(0)+e^{-t/P(0)}[m-P(0)]\geq m$.
If no tie occurs, pacing gives at least $P(0)$.
When $m=R(0)>P(0)$, the same guarantee follows instead from
\cref{lem:race-baseline}.

At any history ending at $S$, speed limits imply
$\pi_i(S_s)\leq\pi_i(S)e^{\alpha s}$ and
$\pi_j(S_s)\geq\pi_j(S)e^{-\beta s}$.
For $0<d<d'$, common survival gives
$d\int_0^\infty e^{-ds-\Lambda_s}\,\de s
\leq d'\int_0^\infty e^{-d's-\Lambda_s}\,\de s$,
because the two sides equal
$\mathbb E[1-e^{-d\tau_D}]$ and
$\mathbb E[1-e^{-d'\tau_D}]$.
Taking $d=r-\alpha$ and $d'=r+\beta$, then taking suprema
over the same admissible paths, yields
\begin{equation*}\tag{PR}\label{spe:eq:payoff-ratio}
V_j(\sigma,h)\geq
k\,\frac{\pi_j(S)}{\pi_i(S)}V_i(\sigma,h).
\end{equation*}

Normalize the initial safety tie to zero and let firm 1 pause
until $\varepsilon$. Every response ends at
$S_\varepsilon=(0,x,g\varepsilon)$ with $x\in[0,\varepsilon]$.
Restart elapsed time there and put
$H:=T+\Delta=g\varepsilon/(1-g)$.
Firm 2 can follow
$A_{2s}^x=\max\{x,\min\{x+s,g\varepsilon+gs\}\}$
until $H$, then resume equilibrium.
This path waits or races to safety and then paces.
Since $A_{1s}\leq s\leq A_{2s}^x$ for $s\leq H$,
firm 2 remains weakly ahead. Exposure is at most
$\Gamma_\varepsilon(s)$, and at $H$ firm 2 is on safety
with capability $H$, where its continuation guarantee is
$e^{cH}m$.

The comparison path with $x=\varepsilon$ gives firm 2 at least
$B_\varepsilon(m):=\widetilde L_2(\varepsilon)
+e^{-\Gamma_\varepsilon(T)-(r-c)H}[m-R(0)]$.
Moreover,
$0\leq A_{2s}^{\varepsilon}-A_{2s}^x\leq\varepsilon-x$,
so
$e^{-(\alpha+\beta)x}e^{\alpha A_{2s}^x-\beta s}
\geq e^{-(\alpha+\beta)\varepsilon}
e^{\alpha A_{2s}^{\varepsilon}-\beta s}$.
The terminal guarantee is independent of $x$ and satisfies
the same comparison. Subgame perfection for firm 2 and
\cref{spe:eq:payoff-ratio} therefore imply
$V_1(\sigma,h_\varepsilon)
\geq k e^{-(\alpha+\beta)\varepsilon}B_\varepsilon(m)$
after every response to the pause.

During the pause, firm 1's flow is at least $e^{-\beta s}$
and exposure is at most $\Lambda_s^R$.
Pausing and then resuming equilibrium consequently gives at least
$\int_0^\varepsilon e^{-(r+\beta)s-\Lambda_s^R}\,\de s
+k e^{-(r+\alpha+\beta)\varepsilon-\Lambda_\varepsilon^R}
B_\varepsilon(m)
=\widetilde L_1(\varepsilon)+q(\varepsilon)[m-R(0)]$.
By \cref{spe:lem:paths}, each prescribed segment can be joined
to a continuation approaching the relevant firm's assigned
payoff. The initial pacing argument uses the same concatenation
at a first tie along any admissible response, and the argument applies symmetrically after every safety tie.
\end{proof}

\emph{Iterating the guarantee.}
Starting from the racing guarantee, define
\begin{equation*}\tag{RF}\label{spe:eq:recursive-floor}
\begin{aligned}
\underline{V}_0&:=R(0),\\
\underline{V}_{n+1}&:=\max\left\{R(0),\sup_{\varepsilon>0}
\left[\widetilde L_1(\varepsilon)
+q(\varepsilon)(\underline{V}_n-R(0))\right]\right\}.
\end{aligned}
\end{equation*}
The first iterate is \cref{spe:eq:refined-pause}. The limit has the formula
\begin{equation*}\tag{LF}\label{spe:eq:recursive-limit}
\underline{V}_\infty=\max\left\{R(0),\sup_{\varepsilon>0}
\frac{\widetilde L_1(\varepsilon)-q(\varepsilon)R(0)}
     {1-q(\varepsilon)}\right\}.
\end{equation*}

\begin{lemma}[Recursive guarantees tighten the risk bound]
\label{lem:recursive-spe-floor}
If the hypotheses of \cref{thm:spe}(b) hold, then the
floors $\underline{V}_n$ increase weakly to a finite limit $\underline{V}_\infty$, given by \cref{spe:eq:recursive-limit}.
After every feasible history ending at a safety tie, every \SPE{} gives
each firm normalized payoff at least $\underline{V}_\infty\geq \underline{V}_1$.
The inequality is strict whenever $\underline{V}_1>R(0)$.
If $R(0)\leq P(0)$, then $\underline{V}_\infty\leq P(0)$ and
\[
0\leq \underline{V}_\infty-\underline{V}_n
\leq\left(\frac{r-\alpha}{r+\beta}\right)^n
       [P(0)-R(0)].
\]
If $R(0)>P(0)$, then $\underline{V}_n=\underline{V}_\infty=\underline{V}_1=R(0)$ for every $n$.
\end{lemma}

\begin{proof}[Proof of \cref{lem:recursive-spe-floor}]
\emph{Step 1: verify each floor.}
If $R(0)\leq P(0)$, the pause bound maps any valid floor
$\underline{V}_n\leq P(0)$ into the valid floor $\underline{V}_{n+1}$.
The pacing \MPE{} in \cref{thm:main} is an \SPE{} giving exactly $P(0)$
from a safety tie, so $\underline{V}_{n+1}\leq P(0)$.
Induction from $\underline{V}_0=R(0)$ proves the guarantee after
every feasible history ending at a safety tie. If $R(0)>P(0)$, the racing \MPE{} gives
exactly $R(0)$ at a safety tie. The first pause floor must therefore
equal $R(0)$, and all subsequent iterates remain there.

\emph{Step 2: identify the limit and the improvement.}
The update in \cref{spe:eq:recursive-floor} increases with its input.
Changing that input by an amount $h$ changes the output by at most
$|h|(r-\alpha)/(r+\beta)$, since
$0<q(\varepsilon)\leq(r-\alpha)/(r+\beta)$.
Since $\underline{V}_1\geq\underline{V}_0$, the floors increase.
Step 1 bounds them above, so they have a finite limit. The displayed
bound on changes makes the update continuous, so it leaves this limit
unchanged. Applying the same bound to $\underline{V}_n$ and
$\underline{V}_\infty$ at each iteration gives the stated error bound.
Each firm receives at least every finite floor, hence at least the limit.

For any proposed floor $m$, the update is at most $m$ exactly when
$m\geq R(0)$ and $ m\geq\dfrac{\widetilde L_1(\varepsilon)-q(\varepsilon)R(0)}
               {1-q(\varepsilon)}$ for every $\varepsilon>0.$ The limit is unchanged by the update, so these inequalities put it
above the right side of \cref{spe:eq:recursive-limit}. Conversely,
setting $m$ equal to that right side makes the update at most $m$,
so every iterate is at most $m$. This proves the formula.

If $\underline{V}_1>R(0)$, a finite positive
pause attains the original optimum: its payoff is continuous,
tends to $\frac{r-\alpha}{r+\beta}R(0)<R(0)$ as the pause vanishes,
and, as it grows without bound, tends to
\[
\int_0^\infty e^{-(r+\beta)s-\Lambda_s^R}\,\de s
<\int_0^\infty e^{-(r-\alpha+\beta)s-\Lambda_s^R}\,\de s
=R(0).
\]
The continuation term vanishes: ignoring hazard bounds the rival's
comparison payoff by $e^{\alpha\varepsilon}/(r-\alpha)+R(0)$,
while its multiplier decays at least as
$e^{-(r+\alpha+\beta)\varepsilon}$.
At that maximizing pause, replacing $R(0)$ by $\underline{V}_1$ adds the strictly
positive amount $q(\varepsilon)[\underline{V}_1-R(0)]$. Thus $\underline{V}_2>\underline{V}_1$ and
$\underline{V}_\infty>\underline{V}_1$. If $\underline{V}_1=R(0)$, the iteration is constant instead.
\end{proof}

\begin{proof}[Proof of \cref{thm:spe}(b)]
Fix an \SPE{} outcome from a safety tie, and normalize the initial
common capability to zero. At date $s$, both capabilities
lie in $[0,s]$. Their summed flow is convex, so it is maximized at a corner
of this square. The maximizing corner has both capabilities equal
to $s$, or one equal to $s$ and the other zero.  Let $F(s):=\max\left\{e^{(\alpha-\beta)s},
\frac{e^{\alpha s}+e^{-\beta s}}2\right\}.$ This gives the flow ceiling $2F(s)$. Since an \SPE{}
outcome attains both firms' assigned payoffs on the same path,
\cref{lem:recursive-spe-floor} implies
\[
\underline{V}_\infty\leq\frac{V_1(\sigma,h_0)+V_2(\sigma,h_0)}{2}
\leq\int_0^\infty e^{-rs-\Lambda_s}F(s)\,\de s.
\]
Fix a finite date $t$ and write $z=\Lambda_t$. The hazard ceiling
$\lambda((1-g)s)$ and monotonicity of exposure imply
\[
\Lambda_s\geq
\begin{cases}
\max\{0,z-(\Lambda_t^R-\Lambda_s^R)\},&s\leq t,\\
z,&s\geq t.
\end{cases}
\]
Substituting this bound gives $\underline{V}_\infty\leq\overline{V}(t,z)$, where
\begin{equation*}\tag{UC}\label{spe:eq:payoff-ceiling}
\begin{aligned}
\overline{V}(t,z)
&:=\int_0^t e^{-rs}F(s)
\min\{1,e^{-z+\Lambda_t^R-\Lambda_s^R}\}\,\de s\\
&\quad+e^{-z}\int_t^\infty e^{-rs}F(s)\,\de s.
\end{aligned}
\end{equation*}
This comparison delays exposure as long as possible and allows no
further risk after $t$; it need not be a feasible capability path.
For fixed $t$, $\overline V(t,z)$ is continuous and strictly decreasing
in $z$: the first term is nonincreasing, while the positive tail is
strictly decreasing. Define
$
\overline\Lambda(t):=
\max\left\{
z\in[0,\Lambda_t^R]:
\overline V(t,z)\geq\underline V_\infty
\right\}.$ The set is nonempty. At least one of the racing and pacing profiles is
an \SPE{}, and applying the preceding payoff bound to an attaining
outcome of that equilibrium produces some
$z\in[0,\Lambda_t^R]$ for which
$\overline V(t,z)\geq\underline V_\infty$.

Since \(\overline V(t,\Lambda_t)\geq\underline V_\infty\), we have \(\Lambda_t\leq\overline\Lambda(t)\), so e $
\Pr(\tau_D\leq t)=1-e^{-\Lambda_t}
\leq1-e^{-\overline\Lambda(t)}.$ \end{proof}

\subsection{Proof of \cref{prop:high SPE}}
\begin{proof}
Put $\delta:=r-\alpha+\beta>0$ and $q:=\alpha-r-\lambda>0$.
Choose $\rho\in(0,1)$ sufficiently close to one that
\[
r+\beta+2\log\rho>\alpha-r,
\qquad
\frac{\rho^2}{1-\rho^2}
\frac{1-e^{-(r+\beta+\lambda)}}{r+\beta+\lambda}
>\frac1 \delta.
\tag{*}
\]
Such a $\rho$ exists because $r+\beta>\alpha-r$ and the left-hand
side of the second inequality diverges as $\rho\uparrow1$.

Let $M_n:=e^{\alpha-r}\rho^n$ for $n\geq-1$.  Construct a reference
path as follows.  Firms alternate: in even block $n$, firm $1$ races
and firm $2$ waits; in odd block $n$, firm $2$ races and firm $1$
waits.  Let $T_0=0$, and end block $n$ at the first date $T_{n+1}$
at which the active firm's discounted flow
$f_i(t):=e^{-rt-\Lambda_t}\pi_i(S_t)$ reaches $M_n$.

While active, $\de\log f_i/\de t\in[q,\alpha-r]$; while inactive,
$\de\log f_i/\de t\in[-(r+\beta+\lambda),-(r+\beta)]$.  Hence every
target is reached in finite time.  Writing
$\Delta_n:=T_{n+1}-T_n$, we have $\Delta_0\geq1$ and
\[
\Delta_{n+1}\geq
\frac{(r+\beta)\Delta_n+2\log\rho}{\alpha-r}.
\]
The first inequality in (*) implies inductively that $\Delta_n\geq1$
for all $n$, so the switching dates have no finite accumulation point.

Payoffs are finite.  An active block ending at $M_n$ contributes at
most $M_n/q$, and the following inactive block contributes at most
$M_n/(r+\beta)$.  Since each firm's own targets decline geometrically
by the factor $\rho^2$, these bounds are summable.

Fix a date $t\in[T_n,T_{n+1})$.  Both current discounted flows are at
most $M_{n-1}$.  Each firm's next target is at least
$\rho^2M_{n-1}$, and after every future target it waits for at least
one unit of time.  During that unit its discounted flow is at least
$M_m e^{-(r+\beta+\lambda)s}$.  Thus, by the second inequality in (*),
$
\int_t^\infty f_i(s)\,\de s
>
\frac{M_{n-1}}\delta
\geq \frac{\max_j f_j(t)}\delta$ for $i=1,2$.
Restarting discounting and survival at $t$ gives
\[
V_i^*(h_t)>
\frac{\max_j\pi_j(S_t)}{r-\alpha+\beta},
\qquad i=1,2.
\tag{**}
\]

To support this path, prescribe the reference actions as long as the
public history exactly coincides with the reference history.  After
any departure, prescribe perpetual racing forever.  Perpetual racing
is an SPE in the stated hazard range: against a racing opponent, every
measurable deviation yields at most
$\max_j\pi_j(S)/\delta$.  By (**), at every reference history
each firm's prescribed continuation is strictly better than this
off-path bound.  Hence no deviation, nor any alternative admissible
continuation that leaves the reference path, is profitable.  After a
departure, perpetual racing is itself sequentially optimal.  The
strategy profile is therefore an SPE.

Initially, (** ) gives $V_i^*>1/(r-\alpha+\beta)$.  Since $g<1$,
$
\frac1{r-\alpha+\beta}>
\frac1{r-g(\alpha-\beta)}=P.$ 
Thus $V_i^*>P$ for both firms.
\end{proof}

\subsection{Proof of \cref{prop:pacing-iff}}
\begin{proof}
We first bound equilibrium payoffs from below. We then show that
racing attains the bound and use it to sustain pacing.

\emph{Step 1: the payoff guarantee and necessity.}
Fix a feasible history $h_t$ ending at state $S=(A_1,A_2,X)$.
Let firm $i$ race thereafter. Its rival develops no faster than
under mutual racing, so firm $i$'s flow profit is at least as high
and hazard is no greater. This comparison holds on every admissible
path: constant racing retains speed one under regularization, and
paths exist by \cref{spe:lem:paths}. Hence every SPE satisfies
$
V_i(\sigma,h_t)\geq\pi_i(A_i,A_j)R_i(S).$ If pacing is an SPE outcome, each continuation attains its assigned
payoff. Otherwise, a better continuation after the exact history
could be concatenated by \cref{spe:lem:paths} to improve the original
finite payoff. At every safety tie $x\geq A_0$ on the pacing path,
the guarantee therefore gives $P_i(x)\geq R_i(x,x,x)$.
Necessity does not use the hazard ceiling.

\emph{Step 2: racing attains the bound.}
Restart time at any state $S=(A_1,A_2,X)$ and hold the rival to
perpetual racing. For any feasible own path, let
$\delta_t:=A_i+t-A_{it}$ be its capability shortfall relative to
racing. Since speed is at most one, this shortfall is nonnegative
and nondecreasing. With $D=\max\{A_1,A_2\}-X$, the distance above
safety along the alternative path satisfies
\[
D+(1-g)s-\delta_t\leq D_s\leq D+(1-g)s
\qquad(0\leq s\leq t).
\]
These inequalities hold whether firm $i$ initially leads, trails,
or is tied. Writing $\Lambda_t^R$ and $\Lambda_t$ for cumulative
exposure under racing and the alternative, monotonicity of hazard
and a change of variables give
\[
\begin{aligned}
0\leq\Lambda_t^R-\Lambda_t
&\leq\int_0^t
\bigl[\lambda(D+(1-g)s)-\lambda(D+(1-g)s-\delta_t)\bigr]\,\de s\\
&=\frac{1}{1-g}\left[
\int_{D+(1-g)t-\delta_t}^{D+(1-g)t}\lambda(z)\,\de z
-\int_{D-\delta_t}^{D}\lambda(z)\,\de z\right] \leq\frac{\bar\lambda}{1-g}\delta_t.
\end{aligned}
\]
The last inequality drops the second integral and bounds the first
by its length times $\bar\lambda$. Racing raises log flow profit
by at least $\alpha\delta_t$, so
\[
\frac{e^{-\Lambda_t^R}\pi_i(A_i+t,A_j+t)}
{e^{-\Lambda_t}\pi_i(A_{it},A_j+t)}
\geq\exp\!\left[\left(\alpha-
\frac{\bar\lambda}{1-g}\right)\delta_t\right]\geq1.
\]
Discounting and integrating shows that no feasible own path beats
racing. A constant-racing rival retains speed one under joint
regularization, so the bound covers every history-dependent
deviation. Mutual racing is therefore an SPE whose unique path
attains the lower bound for both firms at every state.

\emph{Step 3: construct and verify pacing.}
Suppose the payoff condition in \cref{prop:pacing-iff} holds.
Prescribe
\[
\sigma_i(h_t)=
\begin{cases}
g,&A_{1s}=A_{2s}=X_s=A_0+gs\text{ for every }s\in[0,t],\\
1,&\text{otherwise}.
\end{cases}
\]
The rule is progressively measurable, and the set of histories
following the prescribed path is closed in the stopped-history
metric. Once a history departs, all sufficiently nearby histories
also record a departure. The nondeviating firm's speed is therefore
one under regularization after that history and every extension.
On the prescribed path, regularization retains $(g,g)$, so pacing
is admissible; it may also allow departures, which we now bound.

Fix a history $h_t$ on the prescribed path, any deviation
$\sigma_i'$, and any admissible continuation under
$(\sigma_i',\sigma_j)$. A continuation that never departs gives
the pacing payoff. Otherwise, let $T\geq t$ be the infimum of
departure dates. Continuity gives $A_{1T}=A_{2T}=X_T=A_0+gT$.
Every history strictly after $T$ records a departure, so the rival
races almost everywhere thereafter. Step 2 bounds the deviator's
payoff from $T$ by $\pi_i(X_T,X_T)R_i(X_T,X_T,X_T)$.
Before $T$, profits coincide with pacing and hazard is zero.
Its gain over pacing is thus at most
\[
e^{-r(T-t)}\pi_i(X_T,X_T)
\bigl[R_i(X_T,X_T,X_T)-P_i(X_T)\bigr]\leq0.
\]
The inequality uses the proposition's payoff condition and includes
$T=t$ and equality of the benchmark payoffs. Taking suprema rules
out every strategy deviation. The same bound with
$\sigma_i'=\sigma_i$ bounds the assigned payoff, while joint pacing
attains it for both firms on the same path. After every other
history, the profile prescribes perpetual racing, which Step 2
verifies as an equilibrium. Thus $\sigma$ is an SPE with permanent
joint pacing as an outcome.
\end{proof}

\subsection{Monitoring lags}
\label{sec:deployment-proof}
\begin{proof}[Proof of \cref{prop:deployment-lag}]

\emph{Step 1: normalize payoffs and establish racing.}
Subtract profits over $[0,L]$, which are fixed by the initial history,
and divide by $e^{-rL+(\alpha-\beta)A_0}$. A model developed at date
$s$ is deployed at $s+L$, so the normalized payoff on a path is
\[
\int_0^\infty
\exp\left\{-rs+\alpha(A_{is}-A_0)-\beta(A_{js}-A_0)
-\int_0^s\lambda(A_t-X_t-gL)\,\de t\right\}\,\de s.
\]
Joint pacing gives $P(0)=1/[r-g(\alpha-\beta)]$.

To sustain racing, prescribe speed $1$ at every private history.
These constant strategies have a unique admissible path. Under any
unilateral deviation, the rival still races. Since both firms start
at $A_0$ and speed cannot exceed one, the deviator cannot overtake
its rival. Its choices therefore leave hazard unchanged, while
racing maximizes its profits at every date. Both firms attain the
normalized payoff $R(0)$, which is finite because $r>\alpha-\beta$.
Thus racing is a Nash equilibrium.

\emph{Step 2: compare development paths while holding the rival's path fixed.}
Suppose $\bar\lambda<(1-g)\alpha$. For this comparison, start the clock at the
date when the firm makes its choice, after any feasible history.
Here $A_{i0}$ and $X_0$ denote its capability and safety at that
date. Keep the rival's development path unchanged.
Current choices cannot affect profits or hazard during the next $L$
units of time, so that interval is common to both alternatives.

Compare any development path for firm $i$ with racing, and write
$\delta_s:=A_{i0}+s-A_{is}$ for its capability shortfall at date $s$.
The shortfall is nonnegative and nondecreasing because speed is at
most one.
Let
$
Y_t:=A_{j,t}-X_t-gL
$
be the distance from safety, at its deployment date, of the rival's
date-$t$ model. Under racing, firm $i$'s corresponding distance is
$
D_t^R:=A_{i0}-X_0-gL+(1-g)t,
$
whereas under the comparison path it is
\(D_t^R-\delta_t\), with
\(0\leq\delta_t\leq\delta_s\) for \(t\leq s\). Thus the hazards at
calendar date \(L+t\) under racing and under the comparison path are,
respectively,
$\lambda_{L+t}^R
 =\lambda\bigl(\max\{D_t^R,Y_t\}\bigr)$ and $
\lambda_{L+t}
 =\lambda\bigl(\max\{D_t^R-\delta_t,Y_t\}\bigr).$

For any \(x,y\in\mathbb R\) and \(0\leq d\leq\bar d\), monotonicity of
\(\lambda\) gives
$\lambda(\max\{x,y\})-\lambda(\max\{x-d,y\})
\leq \lambda(x)-\lambda(x-\bar d).$ Applying this inequality with
\((x,d,\bar d)=(D_t^R,\delta_t,\delta_s)\) gives
$
\lambda_{L+t}^R-\lambda_{L+t}
\leq\lambda(D_t^R)-\lambda(D_t^R-\delta_s).$
The difference in cumulative exposure through
calendar date \(L+s\) satisfies
\[
\begin{aligned}
\int_L^{L+s}(\lambda_u^R-\lambda_u)\,\de u
&=\int_0^s(\lambda_{L+t}^R-\lambda_{L+t})\,\de t\ \leq\int_0^s
 [\lambda(D_t^R)-\lambda(D_t^R-\delta_s)]\,\de t\\
&=\frac1{1-g}\left[
 \int_{D_s^R-\delta_s}^{D_s^R}\lambda(x)\,\de x
-\int_{D_0^R-\delta_s}^{D_0^R}\lambda(x)\,\de x
\right] \leq\frac{\bar\lambda}{1-g}\delta_s.
\end{aligned}
\]
The equality uses \(D_t^R=D_0^R+(1-g)t\). The final inequality follows
because \(\lambda\geq0\): drop the second integral and bound the first
by its interval length \(\delta_s\) times \(\bar\lambda\).

Racing raises log profit at deployment date $L+s$ by $\alpha\delta_s$.
After accounting for survival, the ratio of its expected flow profit
to that along the slower path is therefore at least
\[
\exp\left\{\left(\alpha-\frac{\bar\lambda}{1-g}\right)\delta_s\right\}
\ge1.
\]
Integrating over time shows that racing gives at least as much
payoff against the same rival path. This comparison holds after
every feasible history, whatever the rival's unobserved development.
It therefore also holds in expectation under any belief supported
on rival histories consistent with the firm's observations and the
constraints on development. %\Cref{fig:deployment-hazard-bound} illustrates the hazard comparison.

\begin{comment}
\begin{figure}[htbp]
\centering
\includegraphics[width=\textwidth]{figures/cooperating_deployment_hazard_bound.pdf}
\caption{Bounding the hazard saved by slower development.
Panel (a) shades the hazard difference between the racing path and
that path shifted down by $\delta_s$. The change of variables bounds
this area by the shaded strip in panel (b), divided by $1-g$.
The enclosing rectangle has area $\bar\lambda\delta_s$.
The log-profit gain $\alpha\delta_s$ exceeds the resulting hazard
bound $\bar\lambda\delta_s/(1-g)$.}
\label{fig:deployment-hazard-bound}
\end{figure}
\end{comment}

\emph{Step 3: pacing requires $B(L)\le P(0)$.}
Suppose joint pacing is a Nash outcome under a profile $\sigma$.
Both firms' payoffs, defined as suprema over admissible paths, must
then equal the pacing payoff $P(0)$ after normalization.
Let firm $i$ deviate to speed $1$ at every history. We construct an
admissible path that keeps firm $j$ pacing until date $L$.

Before $L$, firm $j$ observes only the initial rival history.
Replacing firm $i$'s unobserved development by racing therefore
leaves firm $j$'s private history unchanged. The same replacement
works for nearby histories used in regularization: keep their
end dates below $L$, retain firm $j$'s history, and replace
firm $i$'s development after zero by racing. As the original histories approach joint pacing, the modified
histories approach the path where $i$ races and $j$ paces. Firm
$j$ prescribes the same speeds at the original and modified histories. Since speed $g$ was allowed for firm
$j$ along the original pacing path, it is still allowed here.
Firm $i$ now prescribes speed $1$ everywhere, so the pair $(1,g)$
is jointly admissible almost everywhere on $[0,L]$. Extend this path
beyond $L$.\footnote{\Cref{spe:lem:paths} applies with the initial
history on $[-L,0]$ held fixed, so this admissible path can be
continued beyond $L$. Because the deviator prescribes speed $1$
at every history, it continues racing on every such extension.}

The rival's speed bound gives
\[
A_{js}\le
\begin{cases}
A_0+gs,&0\le s\le L,\\
A_0+gL+s-L,&s>L.
\end{cases}
\]
The deviator remains at least as advanced as its rival and therefore
determines hazard. Its cumulative hazard in development time is
$\Lambda_s^R$. Substituting the rival's upper capability bound into
Step 1 gives a normalized deviation payoff of at least $B(L)$.
Since the deviation payoff is the supremum over admissible paths,
it is at least as large as the payoff on this path. Thus a Nash
equilibrium with pacing requires $B(L)\le P(0)$. This argument
does not use the hazard bound in part (b).

\emph{Step 4: construct pacing when $B(L)\le P(0)$.}
Maintain $\bar\lambda<(1-g)\alpha$. Prescribe speed $g$ while the
firm has always paced and all its observations of its rival are
consistent with pacing; otherwise prescribe speed $1$. Here pacing
means following $A_{is}=X_s+gL$. Write $\sigma$ for this profile.

Every speed allowed for a nondeviating firm lies in $[g,1]$.
Once the firm has departed from pacing or observed a rival
departure, every sufficiently nearby history also records a departure.
Its speed must then be $1$, even if its capability later returns to the
pacing path. If neither departure has occurred, the strategy
prescribes speed $g$. At positive dates, however, histories with
small own departures are arbitrarily close and prescribe speed $1$,
so regularization allows that speed too.
Thus joint pacing is admissible, but regularization can also allow
paths that leave pacing. We must bound those paths as well as
paths following a strategy deviation.

Fix any pure deviation $\sigma_i'$ and any admissible path under
$(\sigma_i',\sigma_j)$. If the path leaves joint pacing, let
$
t:=\inf\{s\ge0:A_{1s}\ne X_s+gL\ \text{or}\ A_{2s}\ne X_s+gL\}.
$
Both firms have capabilities $A_0+gt$ at $t$, by continuity.
At every date strictly after $t+L$, firm $j$ has either departed
itself or observed a rival departure, so it must race. Before then
its speed is at least $g$. Consequently,
\[
A_{j,t+s}\ge
\begin{cases}
A_0+gt+gs,&0\le s\le L,\\
A_0+gt+gL+s-L,&s>L.
\end{cases}
\]

Hold this rival path unchanged and replace firm $i$'s development
after $t$ by racing. Step 2 shows that this change can only raise
its payoff. The changed path need not follow $\sigma_i'$; its
payoff serves only as an upper bound. The racing firm stays at least as advanced as
its rival, so its own development determines hazard. Its profit is
no higher than against the lower bound on rival capability above.
The normalized continuation payoff is therefore at most $B(L)$.
Including the preceding pacing interval, the total normalized
payoff on the original path is at most
$[
\int_0^t e^{-[r-g(\alpha-\beta)]s}\,\de s
+e^{-[r-g(\alpha-\beta)]t}B(L)
\le P(0).$

If the path never leaves pacing, it gives $P(0)$ directly.

This bound also applies when $\sigma_i'=\sigma_i$. Thus no path
under the proposed profile, and no path under any unilateral
deviation, gives more than $P(0)$. Joint pacing attains $P(0)$ for
both firms on the same path. The profile is Nash and sustains
pacing, including when $B(L)=P(0)$.

\emph{Step 5: find the lag cutoff.}
Differentiating the displayed formula for $B(L)$ gives
$B'(L)=\beta(1-g)e^{\beta(1-g)L}\int_L^\infty
 e^{-(r-\alpha+\beta)s-\Lambda_s^R}\,\de s>0.
$
The integrals are finite for every finite lag because
$r>\alpha-\beta$. At zero lag,
$B(0)=R(0)=\frac1{r-\alpha+\beta+\lambda^{\mathrm{eff}}}.
$ In part (a), this exceeds $P(0)$, so Step 3 rules out pacing for
every $L>0$.

For the long-lag limit, write
\[
B(L)=\int_0^\infty e^{-(r-\alpha+\beta g)s-\Lambda_s^R}
 e^{-\beta(1-g)\max\{s-L,0\}}\,\de s.
\]
The last factor increases to one with $L$. Thus monotone convergence
gives the limit, possibly infinite,
$
B(\infty):=\lim_{L\to\infty}B(L)
=\int_0^\infty e^{-(r-\alpha+\beta g)s-\Lambda_s^R}\,\de s.
$
Under the hazard bound in part (b), $\Lambda_s^R\le\bar\lambda s$ gives
$
B(\infty)\ge
\int_0^\infty e^{-(r-\alpha+\beta g+\bar\lambda)s}\,\de s>P(0).
$
If $r-\alpha+\beta g+\bar\lambda>0$, the integral is its reciprocal
and exceeds $P(0)$ because
$\bar\lambda<(1-g)\alpha$; otherwise it is infinite.
When $\lambda^{\mathrm{eff}}>(1-g)(\alpha-\beta)$, we have
$B(0)<P(0)$, so continuity and strict monotonicity give the unique
finite root $L^*>0$. At equality, $B(0)=P(0)$ and $L^*=0$.
Steps 3--4 then give part (b).

\emph{Step 6: sustain pacing at high risk.}
Suppose $\lambda^{\mathrm{eff}}\ge(1-g)\alpha$. We first show that
$B(\infty)\le P(0)$. Since $\Lambda_s^R\le\bar\lambda s$, the
definition of $R(0)$ gives $\lambda^{\mathrm{eff}}\le\bar\lambda$.
As hazard converges to $\bar\lambda$, the integrand defining
$B(\infty)$ eventually decays at a positive exponential rate:
$
r-\alpha+\beta g+\bar\lambda
\ge r-g(\alpha-\beta)>0.$ Hence $B(\infty)$ is finite. The definition of effective hazard gives
\[
\int_0^\infty e^{-(r-\alpha+\beta)s}
 [e^{-\Lambda_s^R}-e^{-(1-g)\alpha s}]\,\de s
=R(0)-\frac1{r-\alpha+\beta+(1-g)\alpha}\le0.
\]
Since hazard never falls along a racing path, average hazard
$\Lambda_s^R/s$ is nondecreasing for $s>0$. The difference in
brackets can therefore change sign only from positive to negative. Multiplying by $e^{\beta(1-g)s}$ puts more weight on
the negative part and preserves a nonpositive integral: at a
sign-change date $t$, the multiplier is at most
$e^{\beta(1-g)t}$ before $t$ and at least that large afterward.
If the difference is nonpositive everywhere, the conclusion follows
without splitting the integral. If it is nonnegative everywhere,
the displayed inequality forces it to vanish almost everywhere. Therefore
\[
B(\infty)\le\frac1{r-\alpha+\beta g+(1-g)\alpha}=P(0).
\]

Now prescribe speed $g$ at every private history. These constant
strategies have the unique path $A_{is}=A_0+gs$. Against a rival
that follows this path, any development path for firm $i$ gives
normalized payoff
\[
\int_0^\infty\exp\left\{
 -(r+\beta g)s+\alpha(A_{is}-A_0)
 -\int_0^s\lambda(A_{it}-X_t-gL)\,\de t\right\}\,\de s.
\]
The rival's deployed capability lies on safety, so hazard is zero
whenever firm $i$ is at most as advanced as its rival. The firm
therefore faces the auxiliary problem
of \cref{lem:scalar-control}, with growth coefficient $\eta=\alpha$,
discount rate $\rho=r+\beta g$, and distance from safety
$A_{is}-X_s-gL$. Its racing value from safety is
$R_\eta(0)=B(\infty)\le P_\eta(0)=P(0)$.
Here \(\rho-g\eta=r-g(\alpha-\beta)>0\) and \(\rho-\eta+\bar\lambda>0\). Thus \(\cref{lem:scalar-control-extended}\) bounds every deviation payoff by \(P(0)\). Joint pacing attains this bound for both firms, proving part (c).
\end{proof}

The following extension allows $\rho\leq\eta$ and supplies
the payoff bound used in Step 6 of the deployment-lag proof.

\begin{lemma}[Scalar-control bound with possibly negative net discount]
\label{lem:scalar-control-extended}
Consider the auxiliary scalar-control problem of
\cref{lem:scalar-control}, and suppose that the hazard is bounded, with
$
\bar\lambda:=\sup_{D>0}\lambda(D)<\infty.
$\$
Let \(v:=1-g\), and suppose
$\rho-g\eta>0$ and $\rho-\eta+\bar\lambda>0.$
If
$
R_\eta(0)\leq P_\eta(0)=\dfrac1{\rho-g\eta},$
then every admissible path starting from \(D=0\) has payoff at most
\(P_\eta(0)\). Consequently, pacing on safety is optimal.
\end{lemma}

\begin{proof}
The assumptions ensure that both benchmark values are finite. The
pacing value is finite because \(\rho-g\eta>0\). Moreover,
\(\lambda(D)\uparrow\bar\lambda\) as \(D\to\infty\), so the racing
integrand eventually decays at rate arbitrarily close to
\(\rho-\eta+\bar\lambda>0\).

We construct a nonnegative continuation-value bound. First consider
\(D\geq0\). The racing value
\[
R_\eta(D)=\int_0^\infty
 \exp\left\{-(\rho-\eta)t
 -\int_0^t\lambda(D+vs)\,\de s\right\}\de t
\]
is nonincreasing in \(D\), and hence
\[
R_\eta(D)\leq R_\eta(0)\leq P_\eta(0).
\]
Waiting until safety and then pacing also gives no more than pacing
immediately. Indeed, relative to immediate pacing, waiting weakly
reduces development at every date and weakly increases exposure.
Thus
$
P_\eta(D)\leq P_\eta(0)$ for $D\geq0.$

Suppose first that \(R_\eta(0)=P_\eta(0)\). Set
\(F(D)=R_\eta(D)\) for \(D\geq0\). For \(D>0\), the racing equation is
$
vR_\eta'(D)
=(\rho-\eta+\lambda(D))R_\eta(D)-1.
$
The coefficient on speed in the payoff residual is
\(R_\eta'+\eta R_\eta\). Its right limit at safety satisfies
\[
v\bigl[R_\eta'(0+)+\eta R_\eta(0)\bigr]
  =(\rho-g\eta+\lambda_0)R_\eta(0)-1
  =\lambda_0P_\eta(0)\geq0.
\]
As in the proof of \cref{lem:scalar-control}, concavity of
\(\lambda\) implies that \(\log R_\eta\) is convex. Therefore
\(R_\eta'/R_\eta\) is nondecreasing, so
\(R_\eta'+\eta R_\eta\geq0\) for every \(D>0\). Since racing makes
the payoff residual zero, the residual is nonpositive for every
speed \(a\in[0,1]\).

Now suppose that \(R_\eta(0)<P_\eta(0)\). For \(D>0\), define
$
F(D):=\max\{P_\eta(D),R_\eta(D)\}.
$
We verify that these two values have the same single-crossing
structure as in \cref{lem:scalar-control}. Write
$
Q_\eta(D):=P_\eta'(D)+\eta P_\eta(D).
$ Differentiating the waiting equation gives
$gQ_\eta'
=-(\rho+\lambda)Q_\eta+\eta-\lambda'P_\eta$ and $
Q_\eta(0+)=-\frac{\lambda_0}{g}P_\eta(0)\leq0.$ 
Because \(\lambda\) is concave and nondecreasing, \(\lambda'\) is
nonnegative and nonincreasing. Since \(P_\eta\) is nonincreasing,
\(\eta-\lambda'P_\eta\) is nondecreasing and converges to \(\eta\).
The integrating-factor argument from
\cref{lem:scalar-control} therefore shows that \(Q_\eta\) can cross
zero only once, from negative to positive, and is eventually
positive.

Define
\[
M_\eta(D):=
\exp\left\{
\frac{(\rho-\eta)D+\int_0^D\lambda(x)\,\de x}{v}
\right\}.
\]
Since \(\lambda(D)\uparrow\bar\lambda\),
$
\lim_{D\to\infty}\dfrac{\log M_\eta(D)}D
=\dfrac{\rho-\eta+\bar\lambda}{v}>0,
$
and hence \(M_\eta(D)\to\infty\). Both \(P_\eta(D)\) and
\(R_\eta(D)\) are bounded above by \(P_\eta(0)\), so
$
\dfrac{P_\eta(D)-R_\eta(D)}{M_\eta(D)}
\longrightarrow0.$ 
Subtracting the racing and waiting equations gives
$
\left(\dfrac{P_\eta-R_\eta}{M_\eta}\right)'
=\dfrac{Q_\eta}{vM_\eta}.$ The ratio starts positive, its derivative can change sign only from
negative to positive, and it converges to zero. It therefore crosses
zero exactly once, from positive to negative. The verification
argument in \cref{lem:scalar-control} then gives
$
1+(a-g)F'(D)
  +(\eta a-\rho-\lambda(D))F(D)\leq0$ at every differentiability point \(D>0\) and every \(a\in[0,1]\).
The same inequality holds almost everywhere along paths that spend
time at the crossing.

For \(D<0\), let \(F(D)\) be the value of racing to safety and then
pacing. The coefficient on speed in its payoff residual is
$
\dfrac{\eta\bigl[1-e^{(\rho-\eta)D/v}\bigr]}{\rho-\eta}\geq0,
$ with the continuous interpretation \(-\eta D/v\) when
\(\rho=\eta\). Hence speed one maximizes the residual below safety
and makes it zero. Set \(F(0)=P_\eta(0)\); pacing makes the residual
zero at safety. The candidate \(F\) is continuous and locally
Lipschitz, and the residual inequality holds almost everywhere along
every admissible path.

Let \(D_t\) be any admissible path starting at \(D_0=0\), and let
\[
Z_t:=\exp\left\{
-\rho t+\eta\int_0^t a_s\,\de s
-\int_0^t\lambda(D_s)\,\de s
\right\}.
\]
Multiplying the residual inequality by \(Z_t\) and integrating through
date \(T\),  $\int_0^T Z_t\,\de t+Z_TF(D_T)\leq F(0)=P_\eta(0).
$ Because \(F\geq0\), dropping the terminal term and letting
\(T\to\infty\) yields
$
\int_0^\infty Z_t\,\de t\leq P_\eta(0).
$
Pacing on safety attains \(P_\eta(0)\), so it is optimal.
\end{proof}

\end{document}

%% file: figures/cooperating_v3_proof_map.tex
% Proof dependencies for the equilibrium construction.
\begin{tikzpicture}[x=1cm,y=1cm,>=Latex,
  block/.style={draw,rounded corners=2pt,align=center,font=\footnotesize,
                inner sep=5pt,text width=4.8cm},
  dep/.style={->,semithick}]
  \node[block] (tied) at (0,0)
    {\cref{lem:scalar-control}\\Tied value and\\lower cutoff};
  \node[block] (race) at (6.5,0)
    {\cref{lem:racing-cutoffs}\\Upper cutoff and\\racing equilibrium};
  \node[block] (bound) at (0,-2.1)
    {\cref{lem:tied-bound}\\Continuation bound for a firm\\tied or behind};
  \node[block] (stopping) at (6.5,-2.1)
    {Steps 1--3: stopping rule\\Stopping region shrinks with the gap;\\waiting preserves it};
  \node[block] (leader) at (0,-4.4)
    {Step 4: leader incentives\\Verify all controls, including\\switches and boundary paths};
  \node[block] (follower) at (6.5,-4.4)
    {Step 5: follower incentives\\Compare every deviation\\with the prescribed path};
  \node[block,text width=10.5cm] (result) at (3.25,-6.8)
    {\cref{lem:pacing-verification,thm:main}\\One path gives both payoffs and no deviation gains;\\the two cutoffs give the three equilibrium regimes};
  \draw[dep] (tied.east)--(race.west);
  \draw[dep] (tied.south)--(bound.north);
  \draw[dep] (tied.south east)--(stopping.north west);
  \draw[dep] (race.south)--(stopping.north);
  \draw[dep] (bound.south)--(leader.north);
  \draw[dep] (bound.south east)--(follower.north west);
  \draw[dep,preaction={draw=white,line width=3pt}]
    (stopping.south west)--(leader.north east);
  \draw[dep] (stopping.south)--(follower.north);
  \draw[dep] (leader.east)--(follower.west);
  \draw[dep] (leader.south)--(result.north west);
  \draw[dep] (follower.south)--(result.north east);
  \draw[dep] (race.east)--(9.45,0)--(9.45,-6.8)--(result.east);
\end{tikzpicture}